\documentclass[conference,10pt]{IEEEtran}

\usepackage[english]{babel}
\usepackage{amsmath}
\usepackage{amssymb}
\usepackage{mathtools}
\usepackage{xspace}
\usepackage{stmaryrd}
\usepackage{url}
\usepackage[bookmarks=false]{hyperref}
\usepackage{pdfrender}

\usepackage{xypic}
\usepackage[all]{xy}
\xyoption{all}
\xyoption{2cell}
\UseTwocells
\xyoption{tips}
\SelectTips{cm}{}  
\usepackage{tikzit}
\newif\ifexternalizetikz

\newif\ifignore 
\ignorefalse

\newcommand{\auxproof}[1]{
\ifignore\mbox{}\newline
\textbf{PROOF:} \dotfill\newline
{\it #1}\mbox{}\newline
\textbf{ENDPROOF}\dotfill
\fi}

\newcommand*{\fatten}[1][.4pt]{%
  \textpdfrender{
    TextRenderingMode=FillStroke,
    LineWidth={\dimexpr(#1)\relax},
  }%
}

\ifdefined\mathsl\else
  \DeclareMathAlphabet{\mathsl}{\encodingdefault}{\rmdefault}{\mddefault}{\sldefault}
  \SetMathAlphabet{\mathsl}{bold}{\encodingdefault}{\rmdefault}{\bfdefault}{\sldefault}
\fi

\makeatletter
\newcommand{\mathoverlap}[2]{\mathpalette\mathoverlap@{{#1}{#2}}}
\newcommand{\mathoverlap@}[2]{\mathoverlap@@{#1}#2}
\newcommand{\mathoverlap@@}[3]{\ooalign{$\m@th#1#2$\crcr\hidewidth$\m@th#1#3$\hidewidth}}
\makeatother

\newcommand{\NNO}{\mathbb{N}}

\newcommand{\setin}[3]{\{#1\in#2\;|\;#3\}}

\newcommand{\andthen}{\mathrel{\&}}

\newcommand{\ket}[1]{\ensuremath{|{\kern.1em}#1{\kern.1em}\rangle}}
\newcommand{\bigket}[1]{\ensuremath{\big|{\kern.1em}#1{\kern.1em}\big\rangle}}
\newcommand{\bra}[1]{\ensuremath{\langle{\kern.1em}#1{\kern.1em}|}}
\newcommand{\upsum}[1]{\ensuremath{#1{\kern-.6ex}\uparrow}\xspace}
\newcommand{\downsum}[1]{\ensuremath{#1{\kern-.6ex}\downarrow}\xspace}

\newcommand{\drawdelete}{\ensuremath{\mathsl{DD}}}

\newcommand{\Scottint}[1]{[{\kern-.3ex}[\,#1\,]{\kern-.3ex}]}
\newcommand{\Lone}[1]{\ensuremath{\|{\kern.3ex}#1{\kern.3ex}\|}_{1}}

\newcommand{\bigScottint}[1]{\big[{\kern-.3ex}\big[\,#1\,\big]{\kern-.3ex}\big]}
\newcommand{\BigScottint}[1]{\Big[{\kern-.3ex}\Big[\,#1\,\Big]{\kern-.3ex}\Big]}

\newcommand{\likely}{\ensuremath{\mathop{|{\kern-1ex}\equiv}}}
\newcommand{\unlikely}{\ensuremath{\mathbin{\smash{|{\kern-0.4ex}{\underset{\raisebox{0.6ex}{\smash{$\scriptstyle u$}}}{\equiv}}}}}}
\newcommand{\nlikely}{\ensuremath{\mathbin{\smash{|{\kern-0.4ex}{\underset{\raisebox{0.6ex}{\smash{$\scriptstyle n$}}}{\equiv}}}}}}
\newcommand{\concat}{\ensuremath{\mathbin{+{\kern-.5ex}+}}}

\newcommand{\supp}{\ensuremath{\mathsl{supp}}}

\newcommand{\facto}[1]{\ensuremath{#1{\kern-2.5pt}\raisebox{-2.5pt}{\includegraphics[width=0.9em]{Pictures/exclamation}}}}
\newcommand{\sfacto}[1]{\ensuremath{#1{\kern-1.5pt}\raisebox{-1.5pt}{\includegraphics[width=0.6em]{Pictures/exclamation}}}}
\newcommand{\coefm}[1]{\ensuremath{\fatten[0.6pt]{(}#1\fatten[0.6pt]{)}}}

\newcommand{\acc}{\ensuremath{\mathsl{acc}}}
\newcommand{\arr}{\ensuremath{\mathsl{arr}}}
\newcommand{\flrn}{\ensuremath{\mathsl{Flrn}}}

\newcommand{\probproj}{\ensuremath{\mathsl{ppr}}}

\newcommand{\iid}{\ensuremath{\mathsl{iid}}}

\newcommand{\zip}{\ensuremath{\mathsl{zip}}}
\newcommand{\mzip}{\ensuremath{\mathsl{mzip}}}

\newcommand{\multinomial}[1][]{\ensuremath{\mathsl{mn}[#1]}}
\newcommand{\pml}{\ensuremath{\mathsl{pml}}}

\newcommand{\hypergeometric}[1][]{\ensuremath{\mathsl{hg}[#1]}}

\newcommand{\intd}{{\kern.2em}\mathrm{d}{\kern.03em}}

\newcommand{\after}{\mathrel{\circ}}
\newcommand{\klafter}{\mathbin{\mathoverlap{\circ}{\cdot}}}

\newcommand{\idmap}[1][]{\ensuremath{\mathrm{id}_{#1}}}

\newcommand{\Category}[1]{\ensuremath{\mathbf{#1}}\xspace}

\newcommand{\Sets}{\Category{Sets}}

\newcommand{\CMon}{\Category{CMon}}

\newcommand{\Kl}{\mathcal{K}{\kern-.4ex}\ell}

\newcommand{\EM}{\mathcal{E}{\kern-.4ex}\mathcal{M}}

\newcommand{\Pow}{\ensuremath{\mathcal{P}}}

\newcommand{\Mlt}{\ensuremath{\mathcal{M}}}
\newcommand{\natMlt}{\ensuremath{\mathcal{M}}}

\newcommand{\neMlt}{\ensuremath{\Mlt_{*}}}

\newcommand{\Dst}{\ensuremath{\mathcal{D}}}

\DeclareSymbolFont{T1op}{T1}{cmr}{m}{n}
\SetSymbolFont{T1op}{bold}{T1}{cmr}{bx}{n}
\DeclareMathSymbol{\mathguilsinglleft}{\mathopen}{T1op}{'016}
\DeclareMathSymbol{\mathguilsinglright}{\mathclose}{T1op}{'017}
\newcommand{\klin}[1]{\mathguilsinglleft#1\mathguilsinglright}

\newcommand{\congrightarrow}{\mathrel{\smash{\stackrel{
           \raisebox{.5ex}{$\scriptstyle\cong$}}{
           \raisebox{0ex}[0ex][0ex]{$\rightarrow$}}}}}

\newtheorem{theorem}{\textbf{Theorem}}
\newtheorem{lemma}[theorem]{\textbf{Lemma}}
\newtheorem{proposition}[theorem]{\textbf{Proposition}}
\newtheorem{corollary}[theorem]{\textbf{Corollary}}

\newtheorem{example}[theorem]{Example}
\newtheorem{remark}[theorem]{Remark}
\newenvironment{proof}[1][Proof]%
   { \begin{trivlist}%
     \item[\hskip \labelsep {\bfseries #1}]%
   }%
   { \end{trivlist}%
   }

\newcommand{\QEDbox}{\square}
\newcommand{\QED}{\hspace*{\fill}$\QEDbox$}

\usetikzlibrary{shapes,shapes.geometric}

\newsavebox\sbpto
\savebox\sbpto{\begin{tikzpicture}[baseline=-2.5pt]
            \filldraw[fill=white,draw=white] circle (1.4pt);
            \filldraw[fill=white,draw=black,line width=0.2pt]circle(2pt);
                \end{tikzpicture}}
\newcommand\chanto{\mathrel{\ooalign{$\to$\cr
      \hfil\!$\usebox\sbpto$\hfil\cr}}}

\IEEEoverridecommandlockouts
\IEEEpubid{\makebox[\columnwidth]{978-1-6654-4895-6/21/\$31.00~
\copyright2021 IEEE \hfill} \hspace{\columnsep}\makebox[\columnwidth]{ }}

\begin{document}

\title{From Multisets over Distributions \\ to Distributions over 
   Multisets
}

\author{\IEEEauthorblockN{Bart Jacobs}
\IEEEauthorblockA{Institute for Computing and Information Sciences,
Radboud University Nijmegen, The Netherlands. \\
Web address: \texttt{www.cs.ru.nl/B.Jacobs}  \\
Email: \texttt{bart@cs.ru.nl}
}}

\maketitle

\IEEEpubidadjcol

\begin{abstract}
A well-known challenge in the semantics of programming languages is
how to combine non-determinism and probability. At a technical level,
the problem arises from the fact that there is a no distributive law
between the powerset monad and the distribution monad --- as noticed
some twenty years ago by Plotkin. More recently, it has become clear
that there is a distributive law of the multiset monad over the
distribution monad. This article elaborates the details of this
distributivity and shows that there is a rich underlying theory
relating multisets and probability distributions. It is shown that the
new distributive law, called \emph{parallel multinomial law}, can be
defined in (at least) four equivalent ways. It involves putting
multinomial distributions in parallel and commutes with hypergeometric
distributions. Further, it is shown that this distributive law
commutes with a new form of zipping for multisets. Abstractly, this
can be described in terms of monoidal structure for a fixed-size
multiset functor, when lifted to the Kleisli category of the
distribution monad. Concretely, an application of the theory to
sampling semantics is included.
\end{abstract}


\section{Introduction}\label{IntroSec}

Monads are used in the semantics of programming languages to classify
different notions of computation, such non-deterministic,
probabilistic, effectful, with exceptions, etc.,
see~\cite{Moggi91a}. An obvious question is if such notions of
computation can be combined via composition of monads. In general,
monads do not compose, but they do compose in presence of a
distributive law. This involves classic work in category theory, due
to Beck~\cite{Beck69}. This topic continues to receive attention
today, especially because finding such distributive laws is a subtle
and error-prone matter, see for instance~\cite{KlinS18} where it is
shown that, somewhat surprisingly, the powerset monad does not
distribute over itself. See also~\cite{Zwart20,ZwartM19} for wider
investigations.

Sometimes it is also interesting whether a functor distributes over a
monad, for instance in operational semantics, see the pioneering work
of~\cite{TuriP97}, which led to much follow-up work, especially in the
theory of coalgebras~\cite{Jacobs16g}. There, distributive laws are
also used for enhanced coinduction
principles~\cite{Bartels03,BonchiPPR17}.

Here we concentrate on the combination of non-deterministic
computation. Some twenty years ago Gordon Plotkin noticed that the
powerset monad $\Pow$ does not distribute over the probability
distributions monad $\Dst$. He never published this important no-go
result himself. Instead, it appeared in~\cite{Varacca03,VaraccaW06},
with full credits, where variations have been investigated (see
also~\cite{GoyP20}).  The lack of a distributive law between $\Pow$
and $\Dst$ means that there is no semantically clean way to combine
non-deterministic and probabilistic computation.

In contrast, multisets do distribute over probability distributions.
This seems to have been folklore knowledge for some time. It is
therefore hard to give proper credit to this
observation. In~\cite[p.82]{KeimelP17} it is described how the set of
distributions $\Dst(M)$ on a commutative monoid $M$, can itself be
turned into a commutative monoid. This is a somewhat isolated
observation in~\cite{KeimelP17}, but by pushing this approach through
one finds that it means that the distribution monad $\Dst$ on the
category of sets lifts to the category of commutative monoids. Since
the latter category is the category of Eilenberg-Moore algebras of the
multiset monad $\Mlt$, some abstract categorical result tells that
this is equivalent to the existence of a distributive law of monads
$\Mlt\Dst \Rightarrow \Dst\Mlt$. Hence one can say that this law is
implicit in~\cite{KeimelP17}, but without an explicit description. The
existence of this law is also discussed in~\cite{DahlqvistPS18}, but
also without an explicit definition. A further source
is~\cite{DashS20}, where the existence of this law is taken for
granted, since the resulting composite $\Dst\Mlt$ is used as a
monad. We simplify things, since~\cite{DashS20} uses continuous
instead of discrete distributions, but that is not essential. Also
there, no explicit description of the distributive law is given.

This absence of an explicit description of the distributive law
$\Mlt\Dst \Rightarrow \Dst\Mlt$ is understandable, since it is quite
complex. The main contribution of this paper is that it describes the
distributive law in full detail --- actually in four different ways
--- and that it develops a rich theory surrounding this law.

The distributive law $\Mlt\Dst \Rightarrow \Dst\Mlt$ turns a multiset
over distributions into a distribution over multisets. At the heart of
this theory that we develop is the interaction of multisets and
distributions. Recall that a multiset is like a set, except that
elements may occur multiple times. A prime example of a multiset is an
urn, containing multiple coloured balls. If the urn contains three
red, five blue and two green balls, then we shall write it as a
multiset $3\ket{R} + 5\ket{B} + 2\ket{G}$ over the set $\{R,B,G\}$ of
colours. This multiset clearly has size 10. A basic property of the
distributive law is that it preserves size: it turns a $K$-sized
multiset of distributions into a distribution over multisets of size
$K\in\NNO$.

When we draw a handful of balls from an urn, the draw itself may also
be represented as a multiset. The classical multinomial and
hypergeometric distributions assign probabilities to such draws ---
with and without replacement, respectively.  We shall describe them as
distributions on $K$-sized multisets. In fact, we shall describe the
distributive law of this paper as a parallel composition of
multinomial distributions. Therefore we propose to name this law the
\emph{parallel multinomial law}, written as $\pml$.

We shall systematically describe multinomial and hypergeometric
distributions via ``channels'', that is, via Kleisli maps in the
Kleisli category $\Kl(\Dst)$ of the distribution monad $\Dst$. It is
within this category that we formulate many of the key properties of
multinomial and hypergeometric channels, using both sequential and
parallel composition.

For the parallel behaviour of these channels and of the distributive
law we introduce a new ``zip'' operation on multisets, called
multizip. It turns out to be mathematically well-behaved, in its
interaction with frequentist learning, with multinomial and
hypergeometric channels, and also with the new parallel multinomial
law. Technically, multizip makes the fixed-sized multiset functor,
lifted to $\Kl(\Dst)$ via the distributive law $\pml$, a monoidal
functor. This multizip is a separate contribution of the paper.

The multinomial and hypergeometric distributions are classics in
probability theory, but they are traditionally not studied as channels
(Kleisli maps), and so their behaviour under sequential and parallel
composition is typically ignored in the literature. This is a pity,
since this behaviour involves beautiful structure, as will be
demonstrated in this paper.

The paper is organised as follows. It first describes the relevant
basics of multisets and probability distributions in
Sections~\ref{MltSec} and~\ref{DstSec}, including the `frequentist
learning' map between them.  Section~\ref{AccArrSec} introduces the
operations of accumulation and arrangement for going back and forth
between sequences and multisets. These two operations play a
fundamental role, later on in the paper, in the description of
multinomial channels, the parallel multinomial role, and also of
multizip. Subsequent Sections~\ref{MulnomSec} and~\ref{HyperGeomSec}
describe multinomial and hypergeometric distributions as channels
(Kleisli maps), together with a number of basic results, expressed via
commuting diagrams, in the Kleisli category of the distribution monad
$\Dst$. This is much like in~\cite{JacobsS20,Jacobs21a}.  The real
work starts in Section~\ref{ParMulnomLawDefSec} where the parallel
multinomial law $\pml$ is introduced in four different, but
equivalent, ways. Basic (sequential, via Kleisli composition)
properties of this law are given in Section~\ref{ParMulnomLawPropSec},
including commutation with multinomial and geometric channels, and
with frequentist learning.

At this point the paper turns to putting things in parallel, again, in
the Kleisli category of the distribution monad $\Dst$. For this
purpose, Section~\ref{ZipSec} introduces a new probabilistic ``zip''
for multisets. It combines two multisets of size $K$, on sets $X,Y$
respectively, into a distribution on $K$-sized multisets on $X\times
Y$. This multiset-zip (or simply: mulitzip, written as $\mzip$) is
thus a channel, and as such interacts well, not only with multinomial
and geometric channels, but also with the parallel multinomial law
$\pml$, see Section~\ref{MzipProbSec}.

In the end, the newly developed theory of multisets and distributions
is applied to sampling semantics. Section~\ref{SampleSec} illustrates
the correctness of the standard sampling approach for parallel and
sequential composition (with a distribution obtained via Kleisli
extension, \textit{i.e.}~state transformation) and for updating
(conditioning of states). These correctness results rely on the
commutation of frequentist learning and multizip with (parallel)
multinomials and on the commutation of multinomials with updating.

\section{Multisets}\label{MltSec}

A multiset is a finite `set' in which elements may occur multiple
times.  We shall write $3\ket{a} + 2\ket{b}$ for a multiset in which
the element $a$ occurs three times and the element $b$ two times. The
ket notation $\ket{-}$ is meaningless syntax that is used to keep
elements (here: $a,b$) and their multiplicities (the numbers: $3,2$)
apart. We shall write $\Mlt(X)$ for the set of finite multisets
$n_{1}\ket{x_1} + \cdots + n_{\ell}\ket{x_{\ell}}$ over the set $X$,
where $x_{i}\in X$ and $n_{i}\in\NNO$. Such a formal sum can also be
written in functional form, as function $\varphi\colon X
\rightarrow\NNO$, whose support $\supp(\varphi) =
\setin{x}{X}{\varphi(x) \neq 0}$ is finite.  The \emph{size} of a
multiset is its total number of elements, written as $\|\varphi\| =
\sum_{x}\varphi(x)$. Thus, $\big\|\,3\ket{a} + 2\ket{b}\,\big\| = 5$.

For multisets $\varphi\in\Mlt(X)$ and $\psi\in\Mlt(Y)$ we can form
their parallel product $\varphi\otimes\psi\in\Mlt(X\times Y)$ via
the product of multiplicities:
\[ \begin{array}{rcl}
\big(\varphi\otimes\psi\big)(x,y)
& \coloneqq &
\varphi(x)\cdot\psi(y).
\end{array} \]

\noindent For instance, $\big(3\ket{a} + 2\ket{b} + 1\ket{c}\big)
\otimes \big(2\ket{0} + 4\ket{1}\big) = 6\ket{a,0} + 12\ket{a,1} +
4\ket{b,0} + 8\ket{b,1} + 2\ket{c,0} + 4\ket{c,1}$.

Taking multisets is functorial: for a function $f\colon X \rightarrow
Y$ one defines $\Mlt(f) \colon \Mlt(X) \rightarrow \Mlt(Y)$ as
$\Mlt(f)(\sum_{i}n_{i}\ket{x_i}) \coloneqq
\sum_{i}n_{i}\ket{f(x_{i})}$.  It is easy to see that
$\|\Mlt(f)(\varphi)\| = \|\varphi\|$.

The functor $\Mlt\colon\Sets\rightarrow\Sets$ is actually a monad,
with unit $\eta\colon X \rightarrow \Mlt(X)$ given by $\eta(x) =
1\ket{x}$. The multiplication $\mu\colon \Mlt^{2}(X) \rightarrow
\Mlt(X)$ is $\mu(\sum_{i}n_{i}\ket{\varphi_{i}})(x) = \sum_{i}
n_{i}\cdot \varphi_{i}(x)$. The Eilenberg-Moore algebras of this monad
are precisely the commutative monoids: for a commutative monoid
$(M,0,+)$ the corresponding algebra $\Mlt(M) \rightarrow M$ sends
formal sums to actual sums: $\sum_{i}n_{i}\ket{m_{i}} \mapsto \sum_{i}
n_{i}\cdot m_{i}$, for $m_{i}\in M$.

We often use multisets of a fixed size $K\in\NNO$, and so we write
$\Mlt[K](X) \subseteq \Mlt(X)$ for the subset of $\varphi\in\Mlt(X)$
with $\|\varphi\| = K$. This $\Mlt[K]$ is also a functor, but not a
monad. As an aside, when the set $X$ has $N$ elements, then
$\Mlt[K](X)$ has $\big(\binom{N}{K}\big) \coloneqq \binom{N+K-1}{K}$
elements, where $\big(\binom{N}{K}\big)$ is called the multiset number
(also known as multichoose).

\section{Probability distributions}\label{DstSec}

In this paper, `distribution' means `finite discrete probability
distribution'. Such a distribution is a formal convex finite sum
$\sum_{i} r_{i}\ket{x_i}$ of elements $x_{i}$ from some set $X$, and
probabilities $r_{i}\in [0,1]$ with $\sum_{i}r_{i} = 1$. We write
$\Dst(X)$ for the set of such distributions over $X$. The mapping $X
\mapsto \Dst(X)$ forms a monad on $\Sets$, just like $\Mlt$. We shall
write $\Kl(\Dst)$ for the Kleisli category of the monad $\Dst$.

A Kleisli map $f\colon X \rightarrow \Dst(Y)$ will also be written as
$f\colon X \chanto Y$ using a special arrow $\chanto$. For a
distribution $\omega\in\Dst(X)$ we write $f \gg \omega \in \Dst(Y)$
for the distribution obtained by Kleisli extension (state
transformation), where:
\[ \begin{array}{rcccl}
\big(f \gg \omega\big)(y)
& = &
\mu\Big(\Dst(f)(\omega)\Big)(y)
& = &
\displaystyle\sum_{x\in X}\, \omega(x)\cdot f(x)(y).
\end{array} \]

\noindent Composition in $\Kl(\Dst)$ will be written as $\klafter$ and
can be described as $(g\klafter f)(x) = g \gg f(x)$. An arbitrary
\emph{function} $h\colon X \rightarrow Y$ can be promoted to a
\emph{channel} $\klin{h} \coloneqq \eta \after h \colon X \chanto
Y$. We frequently promote implicitly, when needed, and omit the
brackets $\klin{-}$.

There is a natural transformation $\flrn \colon \neMlt \Rightarrow
\Dst$ which we call frequentist learning, since it involves learning
by counting. It amounts to normalisation:
\[ \begin{array}{rclcrcl}
\flrn\!\left({\displaystyle\sum}_{i}\,n_{i}\ket{x_i}\right)
& \coloneqq &
{\displaystyle\sum}_{i}\, \displaystyle\frac{n_i}{n}\ket{x_{i}}
& \quad\mbox{where}\quad &
n 
& = & 
\sum_{i}n_{i}.
\end{array} \]

\noindent This $\flrn$ is defined on non-empty multisets, given by
$\neMlt$. It is \emph{not} a map of monads.

For two distributions $\omega\in\Dst(X)$ and $\rho\in\Dst(Y)$ we write
$\omega\otimes\rho \in \Dst(X\times Y)$ for the product distribution.
It is given by $(\omega\otimes\rho)(x,y) = \omega(x)\cdot\rho(y)$,
like for multisets. This tensor $\otimes$ can be extended to channels
as $(f \otimes g)(x,y) = f(x) \otimes g(y)$.  In this way the Kleisli
category $\Kl(\Dst)$ becomes symmetric monoidal. The tensor comes with
projections, since the monad $\Dst$ is affine: $\Dst(1) = 1$, see
\textit{e.g.}~\cite{Jacobs18c} for details.

For a fixed number $K\in\NNO$ we will use a `big' tensor $\bigotimes
\colon \Dst(X)^{K} \rightarrow \Dst(X^{K})$ given by
$\bigotimes(\omega_{1}, \ldots,\omega_{K}) =
\omega_{1}\otimes\cdots\otimes\omega_{K}$. This $\bigotimes$ is
natural in $X$ and commutes appropriately with $\eta$ and $\mu$,
making it a distributive law of the functor $(-)^{K}$ over the monad
$\Dst$. We shall write $\iid \colon \Dst(X) \chanto X^{K}$ for
$\iid(\omega) = \bigotimes(\omega,\ldots,\omega) =
\omega\otimes\cdots\otimes\omega = \omega^{K}$. This $\iid$ gives the
independent and identical distribution.

\section{Accumulation and arrangement}\label{AccArrSec}

There is a simple way to turn a list/sequence of elements into a
multiset, simply by counting occurrences. We call this
\emph{accumulation} and write it as $\acc$. Thus, for instance
$\acc(a,a,b,a) = 3\ket{a} + 1\ket{b}$. More generally we can
simply define:
\[ \begin{array}{rcl}
\acc(x_{1}, \ldots, x_{n})
& \coloneqq &
1\ket{x_1} + \cdots + 1\ket{x_n}.
\end{array} \]

\noindent This accumulation is mathematically well-behaved. It forms
for instance a map of monad, from lists to multisets.

A natural question that arises is: how many lists accumulate to a
given multiset $\varphi$? The (standard) answer is given by what we
call the \emph{multiset coefficient} of the multiset $\varphi$ and
write as $\coefm{\varphi}$. It is:
\[ \begin{array}{rcccl}
\coefm{\varphi}
& \coloneqq &
\displaystyle\frac{\|\varphi\|!}{\prod_{x}\varphi(x)!}
& = &
\displaystyle\binom{\|\varphi\|}{\varphi(x_{1}) \,\cdots\, \varphi(x_{n})},
\end{array} \]

\noindent The latter $\big(-\big)$ notation is common for a
multinomial coefficient, where $\supp(\varphi) = \{x_{1}, \ldots,
x_{n}\}$. For instance, for $\varphi = 2\ket{a} + 3\ket{b}$ there are
$\coefm{\varphi} = \frac{5!}{2!\cdot 3!} = 10$ sequences with length 5
of $a$'s and $b$'s that accumulate to $\varphi$.

We shall concentrate on accumulation for a fixed size, and then write
it as $\acc[K] \colon X^{K} \rightarrow \natMlt[K](X)$. The parameter
$K\in\NNO$ is omitted when it is clear for the context. This
accumulation map is the coequaliser of all permutation maps $X^{K}
\rightarrow X^{K}$. We shall describe it more concretely.

\begin{lemma}
\label{AccCoeqLem}
For a number $K\in\NNO$, let $f\colon X^{K} \rightarrow Y$ be a
function which is \emph{stable under permutation}: for each
permutation $\pi \colon \{1, \ldots, K\} \congrightarrow \{1,\ldots,
K\}$ one has $f\big(x_{1}, \ldots, x_{K}\big) = f\big(x_{\pi(1)},
\ldots, x_{\pi(K)}\big)$, for all sequences $(x_{1}, \ldots, x_{K})\in
X^{K}$. Then there is a unique function $\overline{f} \colon \natMlt[K](X)
\rightarrow Y$ with $\overline{f} \after \acc = f$, as in:
\[ \xymatrix@R-0.8pc{
X^{K}\ar@{->>}[rr]^-{\acc}\ar@/_1ex/[drr]_-{f} & & 
   \natMlt[K](X)\ar@{..>}[d]^-{\overline{f}}
\\
& & Y
} \eqno{\raisebox{-3em}{$\QEDbox$}} \]
\end{lemma}

\smallskip

In the other direction, for a multiset $\varphi\in\Mlt[K](X)$ we
define $\arr[K](\varphi) \in \Dst(X^{K})$ as the (uniform)
distribution over all sequences that accumulate to $\varphi$. Thus:
\[ \begin{array}{rcl}
\arr[K](\varphi)
& \coloneqq &
\displaystyle\sum_{\vec{x}\in\acc^{-1}(\varphi)} 
   \frac{1}{\coefm{\varphi}}\bigket{\vec{x}}.
\end{array} \]

\noindent This is what we call arrangement. It may also be defined
using that $\acc$ is coequaliser.

\begin{lemma}
\label{AccArrLem}
Consider accumulation and arrangement.
\begin{enumerate}
\item \label{AccArrLemNat} They are natural transformations
  $\acc \colon (-)^{K} \Rightarrow \Mlt[K]$ and $\arr \colon
  \Mlt[K] \Rightarrow \Dst((-)^{K})$.

\item \label{AccArrLemId} $\acc \klafter \arr$ is the identity channel
  $\Mlt[K](X) \chanto \Mlt[K](X)$.

\item \label{AccArrLemTensor} The composite $\arr \klafter \acc \colon
  X^{K} \chanto X^{K}$ yields the uniform distribution of all $K!$
  permutations. It commutes with the big tensor $\bigotimes$, as in
  the following diagram of channels.
\[ \xymatrix@R-0.8pc{
\Dst(X)^{K}\ar[r]|-{\circ}^-{\acc}\ar[d]|-{\circ}_{\bigotimes} & 
   \natMlt[K]\big(\Dst(X)\big)\ar[r]|-{\circ}^-{\arr} &
   \Dst(X)^{K}\ar[d]|-{\circ}^-{\bigotimes}
\\
X^{K}\ar[r]|-{\circ}^-{\acc} & \natMlt[K](X)\ar[r]|-{\circ}^-{\arr} & X^{K}
} \eqno{\raisebox{-3em}{$\QEDbox$}} \]
\end{enumerate}
\end{lemma}

\section{Multinomial distributions}\label{MulnomSec}

We can think of a distribution $\omega\in\Dst(X)$ as an abstract urn,
with $X$ as set of colours for the balls in the urn. The number
$\omega(x)\in [0,1]$ gives the probability of drawing one ball of
colour $x$. We consider drawing with replacement, so that the urn does
not change. When we draw a handful of balls, we wish to know the
probability of the draw. Such a draw, say of size $K$, will be
represented as a multiset $\varphi\in\Mlt[K](X)$. Multinomial
distributions assign probabilities to fixed-size draws. They will be
organised as a channel, like in~\cite{JacobsS20,Jacobs21a}, of the
form:
\[ \xymatrix{
\Dst(X)\ar[rr]^-{\multinomial[K]} & & \Dst\big(\natMlt[K](X)\big).
} \]

\noindent For clarity, we write this channel here as a function, with
the distribution monad $\Dst$ explicitly present in its codomain.  We
shall see that writing it as channel $\Dst(X) \chanto \Mlt[K](X)$, and
composing it as such, allows us to smoothly express various properties
of multinomials. This demonstrates the power of (categorical)
abstraction.

On $\omega\in\Dst(X)$, as abstract urn, the multinomial channel
involves the convex sum over multisets, as draws:
\[ \begin{array}{rcl}
\multinomial[K](\omega)
& \,\coloneqq &
\displaystyle\sum_{\varphi\in\natMlt[K](X)} \coefm{\varphi} \cdot
   \textstyle {\displaystyle\prod}_{x}\, \omega(x)^{\varphi(x)}\,\bigket{\varphi}.
\end{array} \]

\noindent The multiset coefficient $\coefm{\varphi}$ of a
draw/multiset $\varphi$, from Section~\ref{AccArrSec} appears because
the order of drawn elements is irrelevant.

The next result captures some basic intuitions about
multinomials: they are suitably additive and the draws match the
original distribution $\omega$, so learning from them yields $\omega$
itself.

\begin{proposition}
\label{MultinomialSumFlrnProp}
Consider the above multinomial channel.
\begin{enumerate}
\item \label{MultinomialSumFlrnPropSum} Draws can be combined: for
  $K,L\in\NNO$,
\[ \hspace*{-1em}\xymatrix@R-0.8pc@C+4pc{
\Dst(X)\ar[d]_-{\Delta}\ar[r]|-{\circ}^-{\multinomial[K+L]} & 
   \natMlt[K\!+\!L](X)
\\
\Dst(X)\!\times\!\Dst(X)
   \ar[r]|-{\circ}^-{\multinomial[K]\otimes\multinomial[L]} &
   \natMlt[K](X)\!\times\!\natMlt[L](X)\ar[u]|-{\circ}_-{+}
} \]

\item \label{MultinomialSumFlrnPropFlrn} $\big(\flrn \klafter
  \multinomial[K]\big)(\omega) = \omega$. \QED
\end{enumerate}
\end{proposition}

\begin{theorem}
\label{MultinomialIIDAccThm}
Multinomial channels are related to accumulation, arrangement, and
independent \& identical distributions:
\[ \xymatrix@R-0.8pc{
\Dst(X)\ar[rr]|-{\circ}^-{\multinomial[K]}\ar@/_2ex/[drr]|-{\circ}_{\iid[K]} 
   & & \natMlt[K](X)\ar[d]|-{\circ}^-{\arr[K]}
\\
& & X^{K}
} \]

\[ \xymatrix@R-1pc{
\Dst(X)\ar[rr]|-{\circ}^-{\multinomial[K]}\ar@/_2ex/[dr]|-{\circ}_{\iid[K]} 
   & & \natMlt[K](X)
\\
& X^{K}\ar@/_2ex/[ur]|-{\circ}_{\acc[K]}
} \]
\end{theorem}

An immediate consequence of this last diagram is that binomial
channels form a natural transformation $\Dst \Rightarrow \Dst\Mlt[K]$.

\begin{proof}
For $\omega\in\Dst(X)$ and $\vec{x} = (x_{1}, \ldots, x_{K})\in
  X^{K}$,
\[ \begin{array}[b]{rcl}
\lefteqn{\big(\arr \klafter \multinomial[K]\big)(\omega)(\vec{x})}
\\
& = &
\displaystyle\sum_{\varphi\in\natMlt[K](X)}\,
   \arr(\varphi)(\vec{x})\cdot \multinomial[K](\omega)(\varphi)
\\[+1.2em]
& = &
\displaystyle\frac{1}{\coefm{\acc(\vec{x})}}\cdot 
   \multinomial[K](\omega)(\acc(\vec{x}))
\\[+0.7em]
& = &
\displaystyle\frac{1}{\coefm{\acc(\vec{x})}}\cdot\coefm{\acc(\vec{x})}\cdot
   \textstyle {\displaystyle\prod}_{y}\, \omega(y)^{\acc(\vec{x})(y)}
\\[+0.8em]
& = &
{\displaystyle\prod}_{i}\, \omega(x_{i})
\hspace*{\arraycolsep}=\hspace*{\arraycolsep}
\omega^{K}(\vec{x})
\hspace*{\arraycolsep}=\hspace*{\arraycolsep}
\iid(\omega)(\vec{x}).
\end{array} \]

\noindent The second diagram then commutes since $\acc \klafter \arr =
\idmap$. \QED
\end{proof}

We conclude our description of multinomial channels by repeating a
result from~\cite{JacobsS20}. It tells that multinomials form a cone
for an infinite chain of draw-and-delete channels $\drawdelete[K]
\colon \Mlt[K\!+\!1](X) \chanto \Mlt[K](X)$. In~\cite{JacobsS20} this
forms the basis for a re-description of de Finetti's theorem in terms
of a limit in a category of channels. We refer the reader to
\emph{loc.\ cit.} for further information; here we just need the cone
property in Lemma~\ref{DrawDeleteLem}~\eqref{DrawDeleteLemMult}
below. First we define withdrawing a single element from a
multiset/urn as a distribution.
\[ \begin{array}{rcl}
\drawdelete[K](\psi)
& \coloneqq &
\displaystyle\!\sum_{x\in\supp(\psi)}\frac{\psi(x)}{K+1}\bigket{\psi - 1\ket{x}}.
\end{array} \]

\noindent Notice that we sum over elements $x$ in the support of the
multiset/urn $\psi$ of size $K+1$, so that $\psi(x) > 0$. Hence we may
remove $x$ from $\psi$, as indicated by the subtraction $\psi -
1\ket{x}$. It leaves us with a multiset of size $K$. The probability
of drawing $x$ depends on the number of occurrences $\psi(x)$. We may
also write the associated probability $\frac{\psi(x)}{K+1}$ as
$\flrn(\psi)(x)$. For instance:
\[ \begin{array}{rcl}
\drawdelete\big(3\ket{a} + 2\ket{b}\big)
& = &
\frac{3}{5}\bigket{2\ket{a} + 2\ket{b}} + 
   \frac{2}{5}\bigket{2\ket{a} + 1\ket{b}}.
\end{array} \]

It is not hard to see that these $\drawdelete$ maps are natural in
$X$. The key property that we are interested in is that
draw-and-delete commutes with multinomial channels.

\begin{lemma}
\label{DrawDeleteLem}
Draw-and-delete of a single element via the channel $\drawdelete
\colon \Mlt[K\!+\!1](X) \chanto \Mlt[K](X)$ satisfies the following
properties.
\begin{enumerate}
\item \label{DrawDeleteLemMult} Commutation with multinomials:
\[ \xymatrix@R-0.8pc{
\natMlt[K\!+\!1](X)\ar[rr]|-{\circ}^-{\drawdelete} & & \natMlt[K](X)
\\
& \Dst(X)\ar@/^2ex/[ul]|-{\circ}^-{\multinomial[K+1]\quad}
   \ar@/_2ex/[ur]|-{\circ}_-{\;\multinomial[K]}
} \]

\item \label{DrawDeleteLemFlrn} Commutation with frequentist learning:
\[ \xymatrix@R-0.8pc@C+1pc{
\natMlt[K\!+\!1](X)\ar[rr]|-{\circ}^-{\drawdelete}
   \ar@/_2ex/[dr]|-{\circ}_-{\flrn} & &
   \natMlt[K](X)\ar@/^2ex/[dl]|-{\circ}^-{\flrn} & 
\\
& X &
} \]
\end{enumerate}
\end{lemma}

\begin{proof}
The first point is proven in~\cite{JacobsS20}, so we concentrate on
the second one. For $\psi\in\Mlt[K+1](X)$ and $y\in X$,
\[ \begin{array}[b]{rcl}
\lefteqn{\big(\flrn \klafter \drawdelete\big)(\psi)(y)}
\\
& = &
\displaystyle\sum_{\varphi\in\natMlt[K](X)} \drawdelete(\psi)(\varphi) \cdot
   \flrn(\varphi)(y)
\\[+1.3em]
& = &
\displaystyle\sum_{x\in X}\, \frac{\psi(x)}{K+1} \cdot\flrn(\psi - 1\ket{x})(y)
\\[+1.2em]
& = &
\displaystyle \frac{\psi(y)}{K+1}\cdot\frac{\psi(y)-1}{K} +
   \sum_{x\neq y}\, \frac{\psi(x)}{K+1}\cdot\frac{\psi(y)}{K}
\\[+1em]
& = &
\displaystyle\frac{\psi(y)}{K(K+1)}\cdot \left(\psi(y)-1 +
   \sum_{x\neq y} \psi(x)\right)
\\[+1em]
& = &
\displaystyle\frac{\psi(y)}{K(K+1)}\cdot \left(\left(
   \sum_{x} \psi(x)\right)-1\right)
\\[+1em]
& = &
\displaystyle\frac{\psi(y)}{K(K+1)}\cdot ((K+1) - 1)
\\[+0.8em]
& = &
\displaystyle\frac{\psi(y)}{K+1}
\\[+0.7em]
& = &
\flrn(\psi)(y).
\end{array} \eqno{\QEDbox} \]
\end{proof}

\section{Hypergeometric distributions}\label{HyperGeomSec}

Where multinomial distributions describe draws with replacement,
hypergeometric distributions involve actual withdrawals, without
replacement. The urn itself is thus a multiset, say of size $N$, and
withdrawals are possible of multisets of (fixed) size $K$, for $K\leq
N$.  The type of the channel is thus:
\begin{equation}
\label{HypergeometricDiag} 
\xymatrix{
\natMlt[N](X)\ar[rr]^-{\hypergeometric[K]} & & \Dst\big(\natMlt[K](X)\big).
}
\end{equation}

\noindent In the definition below we use a pointwise ordering
$\varphi\leq\psi$, for multisets $\varphi,\psi$ on the same set
$X$. It means that $\varphi(x) \leq \psi(x)$ for each $x\in X$.
We write $\varphi \leq_{K} \psi$ if $\varphi\in\natMlt[K](X)$
and $\varphi\leq\psi$.
\[ \begin{array}{rcl}
\hypergeometric[K]\big(\psi)
& \coloneqq &
\displaystyle\sum_{\varphi\leq_{K}\psi} 
   \frac{\prod_{x}\binom{\psi(x)}{\varphi(x)}}
   {\binom{N}{K}}\, \bigket{\varphi},
\end{array} \]

\noindent where $N = \|\psi\|$. The probabilities involved add up to
one by (a generalisation of) Vandermonde's formula $\binom{M+L}{K} =
\sum_{m\leq M, \ell\leq L, m+\ell = K}
\binom{M}{m}\cdot\binom{L}{\ell}$.

The main result about these hypergeometric distributions in this
section is that they can be obtained by iterated draw-and-delete.  In
a slightly different way it is shown in~\cite{Jacobs21a} that both
multinomial and hypergeometric channels can be obtained via iterated
drawing, with or without replacement, where iteration is described in
terms of a Kleisli composition.

\begin{theorem}
\label{DrawDeleteHypergeometricThm}
For $K,L\in\NNO$, the hypergeometric channel $\natMlt[K\!+\!L](X)
\chanto \natMlt[K](X)$ equals $L$ consecutive draw-and-delete's:
\[ \quad\xymatrix@R-0.8pc@C-2pc{
\llap{$\natMlt$}[K\!+\!L](X)\ar[rrrr]|-{\circ}^-{\hypergeometric[K]}
   \ar@/_1ex/[dr]|-{\circ}_(0.3){\drawdelete} & & & & 
    \natMlt[K](X)
\\
& \hspace*{-2em}\natMlt[K\!+\!L\!-\!1]\rlap{$(X)$}
   \ar@{..>}@/_4ex/[rr]|-{\circ}_-{\underbrace{\scriptstyle\drawdelete \,\klafter\, \cdots \,\klafter\, \drawdelete}_{L-2\text{ times}}} & 
   \hspace*{3em} & 
   \natMlt[K\!+\!1](X)\ar@/_1.4ex/[ur]|-{\circ}_-{\drawdelete} &
} \]
\end{theorem}

\begin{proof}
By induction on $L$. \QED
\end{proof}

From this result, and Lemma~\ref{DrawDeleteLem} about draw-and-delete,
we can deduce many additional facts about hypergeometric
distributions.

\begin{corollary}
\label{DrawDeleteHypergeometricCor}
\begin{enumerate}
\item \label{DrawDeleteHypergeometricCorNat} Hypergeometric
  channels~\eqref{HypergeometricDiag} are natural in $X$.

\item \label{DrawDeleteHypergeometricCorFlrn} Frequentist learning
  from hypergeometric draws is like learning from the urn:
\[ \xymatrix@R-0.8pc{
\natMlt[N](X)\ar[rr]|-{\circ}^-{\hypergeometric[K]}
   \ar@/_1.5ex/[dr]|-{\circ}_-{\flrn} & &
   \natMlt[K](X)\ar@/^1.5ex/[dl]|-{\circ}^-{\flrn}
\\
& X &
} \]

\item \label{DrawDeleteHypergeometricCorComp} Hypergeometric channels
  compose, as in:
\[ \hspace*{-2em}\xymatrix@R-0.8pc@C-0.5pc{
\natMlt[K\!+\!L\!+\!M](X)\ar[rr]|-{\circ}^-{\hypergeometric[K]}
   \ar@/_2ex/[dr]|-{\circ}_-{\hypergeometric[K\!+\!L]\qquad} & &
   \natMlt[K](X)
\\
& \natMlt[K\!+\!L](X)\ar@/_2ex/[ur]|-{\circ}_-{\quad\hypergeometric[K]}
} \]

\item \label{DrawDeleteHypergeometricCorMulnom} Hypergeometric and
  multinomial channels commute:
\[ \hspace*{-1em}\xymatrix@R-0.8pc{
\natMlt[K\!+\!L](X)\ar[rr]|-{\circ}^-{\hypergeometric[K]} & & \natMlt[K](X)
\\
& \Dst(X)\ar@/^2ex/[ul]|-{\circ}^-{\multinomial[K\!+\!L]\quad}
   \ar@/_2ex/[ur]|-{\circ}_-{\;\multinomial[K]}
} \eqno{\raisebox{-3em}{$\QEDbox$}} \]
\end{enumerate}
\end{corollary}

\section{The parallel multinomial law: 4 definitions}\label{ParMulnomLawDefSec}

We have already seen the close connection between multisets and
distributions. This section focuses on a very special `distributivity'
connection that shows how a multiset of distributions can be
transformed into a distribution over multisets. This is a rather
complicated operation, but it is fundamental; it can be described via
a tensor product $\otimes$ of multinomials, and will therefor be
called the \emph{parallel multinomial law}, abbreviated as $\pml$.

This law $\pml$ turns out to be a distributive law, in a categorical
sense. It has popped up in~\cite{KeimelP17,DahlqvistPS18} and is used
in~\cite{DashS20}, for continuous probability, to describe `point
processes' as distributions over multisets. This law satisfies several
elementary properties that combine basic ingredients of probability
theory.

This parallel multinomial law $\pml$ that we are after has the following
type. For a number $K\in\NNO$ and a set $X$ it is a function:
\begin{equation}
\label{MltDstDistributiveDiag}
\xymatrix{
\natMlt[K]\Big(\Dst(X)\Big)\ar[r]^-{\pml} & 
   \Dst\Big(\natMlt[K](X)\Big).
}
\end{equation}

\noindent The dependence of $\pml$ on $K$ and $X$ is left
implicit. Later on we develop a version of $\pml$ without the size
parameter $K$.

This map $\pml$ turns a $K$-element multiset of distributions over $X$
into a distribution over $K$-element multisets over $X$. It is not
immediately clear how to do this. It turns out that there are several
ways to describe $\pml$. This section is devoted solely to defining
this law, in four different manners --- yielding each time the same
result. The subsequent section collects basic properties of $\pml$.

\smallskip

\subsubsection*{\textbf{First definition}}

Since the law that we are after is rather complicated, we start with
an example.

\begin{example}
\label{ParMulnomLawFirstEx}
Let $X = \{a,b\}$ be a set with two distributions
$\omega,\rho\in\Dst(X)$, given by:
\[ \begin{array}{rclcrcl}
\omega
& = &
\frac{1}{3}\ket{a} + \frac{2}{3}\ket{b}
& \qquad\mbox{and}\qquad &
\rho
& = &
\frac{3}{4}\ket{a} + \frac{1}{4}\ket{b}. 
\end{array} \]

\noindent We will define $\pml$ on the multiset of distributions
$2\ket{\omega} + 1\ket{\rho}$ of size $K=3$. The result should be a
distribution on multisets of size $K=3$ over $X$. There are four such
multisets, namely:
\[ 3\ket{a}
\qquad
2\ket{a}+1\ket{b}
\qquad
1\ket{a}+2\ket{b}
\qquad
3\ket{b}. \]

\noindent The goal is to assign a probability to each of them. The
map $\pml$ does this in the following way:
\[ \begin{array}{rcl}
\lefteqn{\pml\big(2\ket{\omega} + 1\ket{\rho}\big)}
\\
& = &
\omega(a)\cdot\omega(a)\cdot\rho(a)\bigket{3\ket{a}}
\\[+0.2em]
& & \; +\,
\Big(\omega(a)\cdot\omega(a)\cdot\rho(b) + \omega(a)\cdot\omega(b)\cdot\rho(a)
\\
& & \qquad +\;
\omega(b)\cdot\omega(a)\cdot\rho(a)\Big)\bigket{2\ket{a}+1\ket{b}}
\\[+0.2em]
& & \; +\,
\Big(\omega(a)\cdot\omega(b)\cdot\rho(b) + \omega(b)\cdot\omega(a)\cdot\rho(b)
\\
& & \qquad +\;
\omega(b)\cdot\omega(b)\cdot\rho(a)\Big)\bigket{1\ket{a}+2\ket{b}}
\\[+0.2em]
& & \; +\,
\omega(b)\cdot\omega(b)\cdot\rho(b)\bigket{3\ket{b}}
\\[+0.3em]
& = &
\frac{1}{3}\cdot\frac{1}{3}\cdot\frac{3}{4}\bigket{3\ket{a}}
\\[+0.2em]
& & \; +\,
\Big(\frac{1}{3}\cdot\frac{1}{3}\cdot\frac{1}{4} +
   \frac{1}{3}\cdot\frac{2}{3}\cdot\frac{3}{4} +
   \frac{2}{3}\cdot\frac{1}{3}\cdot\frac{3}{4}\Big)\bigket{2\ket{a}+1\ket{b}}
\\[+0.4em]
& & \; +\,
\Big(\frac{1}{3}\cdot\frac{2}{3}\cdot\frac{1}{4} +
   \frac{2}{3}\cdot\frac{1}{3}\cdot\frac{1}{4} +
   \frac{2}{3}\cdot\frac{2}{3}\cdot\frac{3}{4}\Big)\bigket{1\ket{a}+2\ket{b}}
\\[+0.4em]
& & \; +\,
\frac{2}{3}\cdot\frac{2}{3}\cdot\frac{1}{4}\bigket{3\ket{b}}
\\[+0.3em]
& = &
\frac{1}{12}\bigket{3\ket{a}} \,+\,
   \frac{13}{36}\bigket{2\ket{a}+1\ket{b}} 
\\[+0.2em]
& & \; +\,
   \frac{4}{9}\bigket{1\ket{a}+2\ket{b}} \,+\,
   \frac{1}{9}\bigket{3\ket{b}}.
\end{array} \]

\noindent Notice that the larger outer brackets $\bigket{-}$ involve a
distribution over multisets, given by the smaller inner brackets
$\ket{-}$.
\end{example}

\smallskip

We now formulate the function $\pml$
from~\eqref{MltDstDistributiveDiag} in general, for arbitrary $K$ and
$X$. It is defined on a multiset $\sum_{i} n_{i}\ket{\omega_{i}}$ with
multiplicities $n_{i}\in\NNO$ satisfying $\sum_{i}n_{i} = K$, and with
distributions $\omega_{i}\in\Dst(X)$.  The number
$\pml\big(\sum_{i}n_{i}\ket{\omega_i}\big)(\varphi)$ describes the
probability of the $K$-element multiset $\varphi$ over $X$, by using
for each element occurring in $\varphi$ the probability of that
element in the corresponding distribution in
$\sum_{i}n_{i}\ket{\omega_i}$.

In order to make this description precise we assume that the indices
$i$ are somehow ordered, say as $i_{1},\ldots,i_{m}$ and use this
ordering to form a product distribution $\bigotimes_{i}
\omega_{i}^{n_i} \in\Dst\big(X^{K}\big)$.
\[ \begin{array}{rcl}
\bigotimes_{i} \omega_{i}^{n_i}
& = &
\underbrace{\omega_{i_1}\otimes\cdots\otimes\omega_{i_1}}_{n_{i_1}\text{ times}}
  \otimes \;\cdots\; \otimes
\underbrace{\omega_{i_m}\otimes\cdots\otimes\omega_{i_m}}_{n_{i_m}\text{ times}}.
\end{array} \]

\noindent Now we formulate the first definition:
\begin{equation}
\label{ParMulnomLawFirstEqn}
\begin{array}{rcl}
\lefteqn{\textstyle\pml\big(\sum_{i}n_{i}\ket{\omega_i}\big)}
\\[+0.5em]
& \coloneqq &
\displaystyle\sum_{\vec{x}\in X^{K}}\, \textstyle
   \Big(\!\bigotimes \omega_{i}^{n_i}\Big)(\vec{x})\,
   \bigket{\acc(\vec{x})}
\\
& = &
\displaystyle\sum_{\varphi\in\natMlt[K](X)}\,
   \left(\sum_{\vec{x}\in\acc^{-1}(\varphi)}\, \textstyle
   \Big(\!\bigotimes \omega_{i}^{n_i}\Big)(\vec{x})\right)\bigket{\varphi}.
\end{array}
\end{equation}

\noindent This formulation has been used in
Example~\ref{ParMulnomLawFirstEx}.

\smallskip

\subsubsection*{\textbf{Second definition}}

There is an alternative formulation of the parallel multinomial law,
using multiple multinomial distributions, put in parallel via a tensor
product $\otimes$. This formulation is the basis for the phrase
`parallel multinomial'.

\begin{equation}
\label{ParMulnomLawSecondEqn}
\begin{array}{rcl}
\lefteqn{\textstyle\pml\big(\sum_{i}n_{i}\ket{\omega_i}\big)}
\\[+0.5em]
& \coloneqq &
\Dst\big(\mbox{\Large+}\big)\Big(\!\bigotimes_{i} 
   \multinomial[n_{i}](\omega_{i})\Big)
\\[+0.5em]
& = &
\displaystyle\!\!\sum_{i,\,\varphi_{i}\in\natMlt[n_i](X)} \textstyle
   \Big(\prod_{i} \multinomial[n_{i}](\omega_{i})(\varphi_{i})\Big)\,
   \bigket{\sum_{i}\varphi_{i}}
\end{array}
\end{equation}

\noindent The sum that we use here as type:
\[ \xymatrix@C-1.7pc{
\natMlt[n_{i_1}](X) \!\times \cdots \times\! \natMlt[n_{i_m}](X)
   \ar[rr]^-{\mbox{\Large+}} & & 
   \natMlt[\underbrace{n_{i_1} \!+\! \cdots \!+\! n_{i_m}}_{K}](X).
} \]

\subsubsection*{\textbf{Third definition}}

Our third definition of $\pml$ is more abstract. It uses the
coequaliser property of Lemma~\ref{AccCoeqLem}. It determines $\pml$
as the unique (dashed) map in:
\begin{equation}
\label{ParMulnomLawThirdEqn}
\vcenter{\xymatrix@R-1pc{
\Dst(X)^{K}\ar[dr]_-{\bigotimes}\ar@{->>}[rr]^-{\acc} & &
   \natMlt[K]\big(\Dst(X)\big)\ar@{..>}[dd]^-{\pml}
\\
& \Dst(X^{K})\ar[dr]_-{\Dst(\acc)}
\\
& & \Dst\big(\natMlt[K](X)\big)
}}
\end{equation}

\noindent There is an important side-condition in
Lemma~\ref{AccCoeqLem}, namely that the composite $\Dst(\acc) \after
\bigotimes \colon \Dst(X)^{K} \rightarrow \Dst\big(\natMlt[K](X)\big)$
is stable under permutations. This is easy to check.

This third formulation of the parallel multinomial law is not very
useful for actual calculations, like in
Example~\ref{ParMulnomLawFirstEx}.  But it is useful for proving
properties about $\pml$, via the uniqueness part of the third
definition.

\smallskip

\subsubsection*{\textbf{Fourth definition}}

For our fourth and last definition we have to piece together some
basic observations.
\begin{enumerate}
\item (From~\cite[p.82]{KeimelP17}) If $M$ is a commutative monoid,
  then so is the set $\Dst(M)$ of distributions on $M$, with sum:
\begin{equation}
\label{DstCMonEqn}
\begin{array}{rcl}
\omega + \rho
& \coloneqq &
\Dst(+)(\omega\otimes\rho)
\\[+0.2em]
& = &
\displaystyle\!\sum_{x,y\in M}\, (\omega\otimes\rho)(x,y)\bigket{x\!+\!y}.
\end{array}
\end{equation}

\item Such commutative monoid structure corresponds to an
  Eilenberg-Moore algebra $\alpha\colon \natMlt(\Dst(M)) \rightarrow
  \Dst(M)$ of the multiset monad $\Mlt$, given by:
\[ \begin{array}{rcl}
\alpha\big(\sum_{i}n_{i}\ket{\omega_i}\big)
& = &
\displaystyle\sum_{\vec{x} \in M^{K}}\, \textstyle
   \big(\!\bigotimes_{i}\omega_{i}^{n_i}\big)(\vec{x})
   \,\bigket{\sum\vec{x}\,}
\end{array} \]

\noindent where $K = \sum_{i}n_{i}$.

\item Applying the previous two points with commutative monoid $M =
  \natMlt(X)$ of multisets yields an Eilenberg-Moore algebra:
\begin{equation}
\label{DstMltAlgDiag}
\xymatrix{
\natMlt\Big(\Dst\big(\natMlt(X)\big)\Big)\ar[r]^-{\alpha} &
   \Dst\big(\natMlt(X)\big)
}
\end{equation}
\end{enumerate}

\noindent We now define:
\begin{equation}
\label{ParMulnomLawFourthEqn}
\xymatrix@C-1pc{
\pml
\coloneqq
\Big(\natMlt\Dst(X)\ar[rr]^-{\natMlt\Dst(\eta)} & &
   \natMlt\Dst\natMlt(X)\ar[r]^-{\alpha} &
   \Dst\natMlt(X)\Big).
}
\end{equation}

\begin{proposition}
\label{ParMulnomLawFirstSecondThirdFourthProp}
The definition of $\pml$ in~\eqref{ParMulnomLawFourthEqn} restricts to
$\natMlt[K](\Dst(X)) \rightarrow \Dst\big(\natMlt[K](X)\big)$, for
each $K\in\NNO$. This restriction is the same $\pml$ as defined
in~\eqref{ParMulnomLawFirstEqn}, \eqref{ParMulnomLawSecondEqn}
and~\eqref{ParMulnomLawThirdEqn}.
\end{proposition}

\begin{proof}
For the equivalence of the first two
formulations~\eqref{ParMulnomLawFirstEqn}
and~\eqref{ParMulnomLawSecondEqn} we use
Theorem~\ref{MultinomialIIDAccThm}. This definition of $\pml$ commutes
with $\acc$ and is thus the same as the third one
in~\eqref{ParMulnomLawThirdEqn}.  Finally, when we unravel the fourth
formulation~\eqref{ParMulnomLawFourthEqn} we get the second
one~\eqref{ParMulnomLawSecondEqn}. \QED
\end{proof}

\section{The parallel multinomial law: properties}\label{ParMulnomLawPropSec}

The next result enriches what we already know: the rectangle on the
right below is added to the (known) one on the left.

\begin{proposition}
\label{ParMulnomLawAccArrProp}
The parallel multinomial law $\pml$ is the unique channel making
both rectangles below commute.
\[ \xymatrix@R-0.8pc@C+1pc{
\Dst(X)^{K}\ar[d]|-{\circ}_-{\bigotimes}\ar[r]|-{\circ}^-{\acc} & 
   \natMlt[K]\big(\Dst(X)\big)\ar@{..>}[d]|-{\circ}_-{\pml}
   \ar[r]|-{\circ}^-{\arr} & 
   \Dst(X)^{K}\ar[d]|-{\circ}^-{\bigotimes}
\\
X^{K}\ar[r]|-{\circ}^-{\acc} & 
   \natMlt[K](X)\ar[r]|-{\circ}^-{\arr} & X^{K}
} \]
\end{proposition}

\begin{proof}
The rectangle on the left is the third formulation of $\pml$
in~\eqref{ParMulnomLawThirdEqn} and thus provides uniqueness.
Commutation of the rectangle of the right follows from a 
uniqueness argument, using that the outer rectangle commutes
by Lemma~\ref{AccArrLem}~\eqref{AccArrLemTensor}:
\[ \begin{array}[b]{rcll}
\bigotimes \klafter\, \arr \klafter \acc
& = &
\arr \klafter \acc \klafter \bigotimes
   \qquad & \mbox{by Lemma~\ref{AccArrLem}~\eqref{AccArrLemTensor}}
\\
& = &
\arr \klafter \pml \klafter \acc
   & \mbox{by~\eqref{ParMulnomLawThirdEqn}.}
\end{array} \eqno{\QEDbox} \]
\end{proof}

This result shows that $\pml$ is squeezed between $\bigotimes$, both
on the left and on the right. This $\bigotimes$ is a distributive
law. We shall prove the same about $\pml$ below.

But first we look at interaction with frequentist learning.

\begin{theorem}
\label{ParMulnomLawFlrnThm}
The distributive law $\pml$ commutes with frequentist learning,
in the sense that for $\Psi\in\natMlt[K]\big(\Dst(X)\big)$,
\[ \begin{array}{rcl}
\flrn \gg \pml(\Psi)
& = &
\mu\big(\flrn(\Psi)\big).
\end{array} \]

\noindent Equivalently, in diagrammatic form:
\[ \xymatrix@R-0.8pc@C+1pc{
\natMlt[K]\big(\Dst(X)\big)\ar[r]|-{\circ}^-{\pml}\ar[d]|-{\circ}_{\flrn}
  & \natMlt[K](X)\ar[d]|-{\circ}^{\flrn}
\\
\Dst(X)\ar[r]|-{\circ}^-{\idmap} & X
} \]

\noindent The channel $\Dst(X)\chanto X$ at the bottom is the identity
function $\Dst(X) \rightarrow \Dst(X)$.
\end{theorem}

\begin{proof}
Let $\Psi = \sum_{i}n_{i}\ket{\omega_{i}}  \in
\natMlt[K]\big(\Dst(X)\big)$.
\[ \hspace*{-4em}\begin{array}[b]{rcl}
\lefteqn{\textstyle\flrn \gg \pml(\Psi)}
\\
& \smash{\stackrel{\eqref{ParMulnomLawSecondEqn}}{=}} &
\flrn \gg \Dst\big(\mbox{\Large+}\big)
   \Big(\!\bigotimes_{i} \multinomial[n_{i}](\omega_{i})\Big)
\\[+0.2em]
& \smash{\stackrel{(*)}{=}} &
{\displaystyle\sum}_{i}\; \frac{n_i}{K}\cdot \big(\flrn \gg
   \multinomial[n_{i}](\omega_{i})\big)
\\[+0.5em]
& = &
{\displaystyle\sum}_{i}\; \frac{n_i}{K}\cdot \omega_{i}
   \qquad\mbox{by Proposition~\ref{MultinomialSumFlrnProp}~\eqref{MultinomialSumFlrnPropFlrn}}
\\[+0.5em]
& = &
\mu\big(\sum_{i} \frac{n_i}{K}\ket{\omega_{i}}\big)
\\[+0.2em]
& = &
\mu\big(\flrn(\Psi)\big).
\end{array} \]

\noindent The marked equation $\smash{\stackrel{(*)}{=}}$ follows from
(a generalisation of) the following fact, for
$\Omega\in\Dst(\natMlt[K](X))$, $\Theta\in\Dst(\natMlt[L](X)$.
\[ \begin{array}[b]{rcl}
\lefteqn{\flrn \gg \Dst(+)(\Omega\otimes\Theta)}
\\[+0.3em]
& = &
\frac{K}{K+L}\cdot \big(\flrn \gg \Omega\big) + 
   \frac{L}{K+L}\cdot \big(\flrn \gg \Theta\big).
\end{array} \eqno{\QEDbox} \]
\end{proof}

Our next result shows that the parallel multinomial law commutes
with hypergeometric channels.

\begin{theorem}
\label{ParMulnomLawHypergeomThm}
The parallel multinomial $\pml$ commutes with draw-and-delete:
\[ \xymatrix@R-0.8pc@C+1pc{
\natMlt[K+1]\big(\Dst(X)\big)\ar[d]|-{\circ}_{\pml}
   \ar[r]|-{\circ}^-{\drawdelete} &
   \natMlt[K]\big(\Dst(X)\big)\ar[d]|-{\circ}^-{\pml}
\\
\natMlt[K+1](X)\ar[r]|-{\circ}^-{\drawdelete} & 
   \natMlt[K](X)
} \]

\noindent and then also with hypergeometrics: for $N\geq K$ one has,
\[ \xymatrix@R-0.8pc@C+2pc{
\natMlt[N]\big(\Dst(X)\big)\ar[d]|-{\circ}_{\pml}
   \ar[r]|-{\circ}^-{\hypergeometric[K]} &
   \natMlt[K]\big(\Dst(X)\big)\ar[d]|-{\circ}^-{\pml}
\\
\natMlt[N](X)\ar[r]|-{\circ}^-{\hypergeometric[K]} & 
   \natMlt[K](X)
} \]
\end{theorem}

\begin{proof}
Showing commutation of $\pml$ with $\drawdelete$ suffices, by
Theorem~\ref{DrawDeleteHypergeometricThm}. We use an auxiliary
probabilistic projection channel $\probproj \colon X^{K+1} \chanto
X^{K}$. It is a uniform distribution over projecting away at each
position.
\[ \begin{array}{rcl}
\lefteqn{\probproj(x_{1}, \ldots, x_{K+1})}
\\[+0.2em]
& \coloneqq &
\displaystyle\sum_{1\leq k\leq K+1} \frac{1}{K\!+\!1}
   \bigket{x_{1}, \ldots, x_{k-1}, x_{k+1}, \ldots, x_{K+1}}.
\end{array} \]

\noindent It is not hard to see that the following diagram
commutes.
\[ \xymatrix@R-0.8pc@C-0.8pc{
\Dst(X)^{K+1}\ar[r]|-{\circ}^-{\bigotimes}\ar[d]|-{\circ}_-{\probproj} & 
X^{K+1}\ar[d]|-{\circ}_{\probproj}\ar[r]|-{\circ}^-{\acc} & 
   \natMlt[K\!+\!1](X)\ar[d]|-{\circ}^-{\drawdelete}\ar[r]|-{\circ}^-{\arr} &
   X^{K+1}\ar[d]|-{\circ}^{\probproj}
\\
\Dst(X)^{K}\ar[r]|-{\circ}^-{\bigotimes} & 
X^{K}\ar[r]|-{\circ}^-{\acc} & \natMlt[K](X)\ar[r]|-{\circ}^-{\arr} & X^{K}
} \]

\noindent In the end we use that $\acc$ is coequaliser (so surjective/epic):
\[ \begin{array}[b]{rcll}
\drawdelete \klafter \pml \klafter \acc
& \smash{\stackrel{\eqref{ParMulnomLawThirdEqn}}{=}} &
\drawdelete \klafter \acc \klafter \bigotimes \qquad
\\
& = &
\acc \klafter \probproj \klafter \bigotimes
\\
& = &
\acc \klafter \bigotimes \klafter\, \probproj
\\
& \smash{\stackrel{\eqref{ParMulnomLawThirdEqn}}{=}} &
\pml \klafter \acc \klafter \probproj
\\
& = &
\pml \klafter \drawdelete \klafter \acc.
\end{array} \eqno{\QEDbox} \]
\end{proof}

\begin{proposition}
\label{ParMulnomLawPlusProp}
The parallel multinomial law commutes with sums of multisets:
\[ \xymatrix@R-0.8pc@C+0.0pc{
\natMlt[K]\big(\Dst(X)\big)\!\times\!\natMlt[L]\big(\Dst(X)\big)
   \ar[d]|-{\circ}_-{+}\ar[r]|-{\circ}^-{\pml\otimes\pml} &
   \natMlt[K](X)\!\times\!\natMlt[L](X)\ar[d]|-{\circ}^-{+}
\\
\natMlt[K\!+\!L]\big(\Dst(X)\big)\ar[r]|-{\circ}^-{\pml} & \natMlt[K\!+\!L](X)
} \]
\end{proposition}

\begin{proof}
Via a suitable generalisation of
Proposition~\ref{MultinomialSumFlrnProp}~\eqref{MultinomialSumFlrnPropSum}. \QED
\end{proof}

\begin{theorem}
\label{ParMulnomLawLawThm}
The parallel multinomial law $\pml$ is a \emph{distributive law} in
several ways.
\begin{enumerate}
\item \label{ParMulnomLawLawThmKl} It is a distributive law of the
  multiset \emph{functor} $\natMlt[K]$ over the distribution monad
  $\Dst$. This means that $\pml$ commutes with the unit and
  multiplication operations of $\Dst$. Equivalently, $\natMlt[K]$ can
  be lifted to a functor $\Kl(\Dst) \rightarrow \Kl(\Dst)$.

\item \label{ParMulnomLawLawThmMlt} It is thereby also a law $\Mlt\Dst
  \Rightarrow \Dst\Mlt$ of \emph{functor} $\Mlt$ over the monad
  $\Dst$, as in the fourth definition~\eqref{ParMulnomLawFourthEqn},
  so that $\Mlt$ also lifts to $\Kl(\Dst)$.

\item \label{ParMulnomLawLawThmEM} In addition, $\pml\colon\Mlt\Dst
  \Rightarrow \Dst\Mlt$ is a distributive law of \emph{monads}, of
  $\Mlt$ over $\Dst$. Equivalently, $\Dst$ can be lifted to a monad
  $\Dst \colon \CMon \rightarrow \CMon$, on the category $\CMon =
  \EM(\Mlt)$ of commutative monoids (Eilenberg-Moore algebras of
  $\Mlt$).
\end{enumerate}
\end{theorem}

\smallskip

See \textit{e.g.}~\cite{Jacobs16g} for a detailed description of the
correspondence between laws and lifting. The last point means that the
composite $\Dst\Mlt\colon\Sets\rightarrow\Sets$ is a monad. A
continuous version of this monad is used in~\cite{DashS20}, as
the monad of point processes.

\begin{proof}
An easy way to prove the first point is to use that $\bigotimes \colon
\Dst(X)^{K} \rightarrow \Dst(X^{K})$ is a distributive law and that
$\acc \colon X^{K} \rightarrow \Mlt[K](X)$ is surjective. The second
point then follows easily. For the third point it is easiest to check
that $\Dst(M)$ is a commutative monoid if $M$ is, using the
formula~\eqref{DstCMonEqn} and that this gives a monad on
$\CMon$. \QED
\end{proof}

\section{Zipping multisets}\label{ZipSec}

So far we have seen (parallel) multinomial and hypergeometric channels
and looked at some basic properties, involving their sequential
behaviour, in the Kleisli category $\Kl(\Dst)$. We now move to their
parallel properties, in terms of tensors $\otimes$. This requires a
new (probabilistic) operation to combine multisets (of the same size)
which we call multizip, and write as $\mzip$. This section
concentrates on this $\mzip$, and the next section will look into its
interaction with the main probabilistic channels of this paper.

In (functional) programming zipping two lists together is a common
operation. Mathematically it takes the form of function $\zip \colon
X^{K} \times Y^{K} \congrightarrow (X\times Y)^{K}$. This section
shows how to do a form of zipping for multisets. This is a new and
very useful operation, that, as we shall see later, commutes with
basic probabilistic mappings.

But first we need to collect two auxiliary properties of 
ordinary zip.

\begin{lemma}
\label{ZipLem}
Zipping commutes with tensors of distributions in the following ways.
\[ \xymatrix@R-0.8pc{
\Dst(X)\!\times\!\Dst(Y)\ar[d]_-{\otimes}\ar[r]|-{\circ}^-{\iid\otimes\iid} &
   X^{K}\!\times\! Y^{K}\ar[d]|-{\circ}^-{\zip}
\\
\Dst(X\!\times\! Y)\ar[r]|-{\circ}^-{\iid} & (X\!\times\! Y)^{K}
} \]
\[ \xymatrix@R-0.8pc@C-0.5pc{
\Dst(X)^{K}\times\Dst(Y)^{K}\ar[rr]|-{\circ}^-{\bigotimes\otimes\bigotimes}
   \ar[d]|-{\circ}_-{\zip} & &
   X^{K}\times Y^{K}\ar[d]|-{\circ}^-{\zip}
\\
\big(\Dst(X)\times\Dst(Y)\big)^{K}\ar[r]|-{\circ}^-{\otimes^{K}} &
   \Dst(X\times Y)^{K}\ar[r]|-{\circ}^-{\bigotimes} & (X\times Y)^{K}
} \]
\end{lemma}

\begin{proof}
By straightforward calculation. \QED
\end{proof}

When we take sizes into account, we see that the sum of two multisets
is a function $\Mlt[K](X)\times\Mlt[L](X) \rightarrow
\Mlt[K\!+\!L](X)$. The product (tensor) has type
$\Mlt[K](X)\times\Mlt[L](Y) \rightarrow \Mlt[K\!\cdot\!L](X\times Y)$.
The multizip function $\mzip$ that we are about to introduce takes
inputs of the same size and does not change this size. Its type is
thus given by the top row of the next diagram. The rest of the diagram
forms its definition:
\[ \xymatrix@R-1.5pc@C-1.5pc{
\natMlt[K](X)\times \natMlt[K](Y)\ar[dd]|-{\circ}_-{\arr\otimes\arr}
    \ar[rr]|-{\circ}^-{\mzip[K]} & &
   \natMlt[K](X\times Y)
\\
& \text{\llap{defi}nition of $\mzip$}
\\
X^{K}\times Y^{K}\ar[rr]|-{\circ}^-{\zip} & & 
   (X\times Y)^{K}\ar[uu]|-{\circ}_-{\acc}
} \]

\noindent Explicitly, for multisets $\varphi\in\natMlt[K](X)$ and
$\psi\in \natMlt[K](Y)$,
\[ \begin{array}{rcl}
\mzip(\varphi,\psi)
& \coloneqq &
\displaystyle\!\!\sum_{\vec{x}\in\acc^{-1}(\varphi)} \sum_{\vec{y}\in\acc^{-1}(\psi)}\!
   \frac{1}{\coefm{\varphi}\cdot\coefm{\psi}}
   \bigket{\acc\big(\zip(\vec{x}, \vec{y})\big)}.
\end{array} \]

\noindent Thus, in zipping two multisets $\varphi$ and $\psi$ we first
look at all their arrangements $\vec{x}$ and $\vec{y}$, zip these
together in the ordinary way, and accumulate the resulting list of
pairs into a multiset again. We elaborate how this works concretely.

\begin{example}
\label{MzipEx}
Let's use two set $X = \{a, b\}$ and $Y = \{0,1\}$ with two multisets
of size three:
\[ \begin{array}{rclcrcl}
\varphi
& = &
1\ket{a} + 2\ket{b}
& \qquad\mbox{and}\qquad &
\psi
& = &
2\ket{0} + 1\ket{1}.
\end{array} \]

\noindent Then:
\[ \begin{array}{rccclcrcccl}
\coefm{\varphi}
& = &
\binom{3}{1,2}
& = &
3
& \qquad &
\coefm{\psi}
& = &
\binom{3}{2,1}
& = &
3
\end{array} \]

\noindent The sequences in $X^3$ and $Y^3$ that accumulate to
$\varphi$ and $\psi$ are:
\[ \left\{\begin{array}{l}
a,b,b \\[-0.3em]
b,a,b \\[-0.3em]
b,b,a 
\end{array}\right.
\qquad\mbox{and}\qquad
\left\{\begin{array}{l}
0,0,1 \\[-0.3em]
0,1,0 \\[-0.3em]
1,0,0 
\end{array}\right. \]

\noindent Zipping them together gives the following nine sequences in
$(X\times Y)^{3}$.
\[ \begin{array}{ccccc}
(a,0), (b,0), (b,1)
& \; &
(b,0), (a,0), (b,1)
& \; &
(b,0), (b,0), (a,1)
\\[-0.3em]
(a,0), (b,1), (b,0)
& &
(b,0), (a,1), (b,0)
& &
(b,0), (b,1), (a,0)
\\[-0.3em]
(a,1), (b,0), (b,0)
& &
(b,1), (a,0), (b,0)
& &
(b,1), (b,0), (a,0)
\end{array} \]

\noindent By applying the accumulation function $\acc$ to each of
these only two multisets remain:
\[ \begin{array}{ccc}
\text{three of } 1\ket{a,1} \!+\! 2\ket{b,0}
& \quad &
\text{six of } 1\ket{a,0} \!+\! 1\ket{b,0} \!+\! 1\ket{b,1}.
\end{array} \]

\noindent Finally, multiplication with
$\frac{1}{\coefm{\varphi}\cdot\coefm{\psi}} = \frac{1}{9}$ gives
the outcome:
\[ \begin{array}{rcl}
\lefteqn{\mzip[3]\Big(1\ket{a} \!+\! 2\ket{b}, \, 2\ket{0} \!+\! 1\ket{1}\Big)}
\\[+0.3em]
& = &
\frac{1}{3}\bigket{1\ket{a,1} \!+\! 2\ket{b,0}} \;+\;
   \frac{2}{3}\bigket{1\ket{a,0} \!+\! 1\ket{b,0} \!+\! 1\ket{b,1}}.
\end{array} \]

\noindent This shows that calculating $\mzip$ is laborious.  But it is
quite mechanical and easy to implement.  The picture below suggests to
look at $\mzip$ as a funnel with two input pipes in which multiple
elements from both sides can be combined into a probabilistic mixture.
\medskip
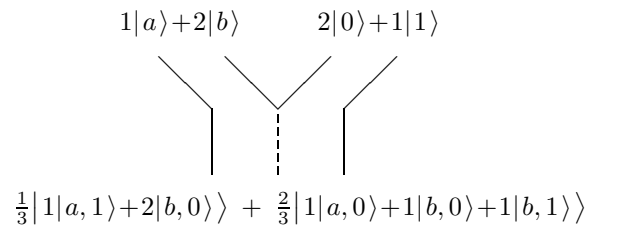
\begin{center}
\begin{picture}(150,75)(0,0)
\put(15,70){$1\ket{a} \!+\! 2\ket{b}$}
\put(90,70){$2\ket{0} \!+\! 1\ket{1}$}
\put(50,15){\line(0,1){25}}
\put(50,40){\line(-1,1){20}}
\put(75,40){\line(-1,1){20}}
\put(75,40){\line(1,1){20}}
\put(100,40){\line(1,1){20}}
\put(100,15){\line(0,1){25}}
\put(75,15){\line(0,1){3}}
\put(75,20){\line(0,1){3}}
\put(75,25){\line(0,1){3}}
\put(75,30){\line(0,1){3}}
\put(75,35){\line(0,1){3}}
\put(-25,0){$\frac{1}{3}\bigket{1\ket{a,1} \!+\! 2\ket{b,0}} \;+\;
   \frac{2}{3}\bigket{1\ket{a,0} \!+\! 1\ket{b,0} \!+\! 1\ket{b,1}}$}
\end{picture}
\end{center}
\medskip
\end{example}

\begin{lemma}
\label{MzipLem}
Consider multizip $\mzip$ as channel $\natMlt[K](X)\times\natMlt[K](Y)
\chanto \natMlt[K](X\times Y)$.
\begin{enumerate}
\item \label{MzipLemNat} Zipping of multisets is a natural
  transformation $\natMlt[K](-)\times\natMlt[K](-) \Rightarrow
  \Dst\natMlt[K](-\times -)$.

\item \label{MzipLemUnit} $\mzip(\varphi, K\ket{y}) =
  1\bigket{\varphi\otimes 1\ket{y}}$.

\item \label{MzipLemAss} $\mzip$ is associative, in $\Kl(\Dst)$.

\item \label{MzipLemProj} $\mzip$ commutes with projections, as in
  $\natMlt[K](\pi_{i}) \klafter \mzip = \pi_{i}$, but not with
  diagonals.

\item \label{MzipLemAcc} Arrangement $\arr$ relates $\zip$ and
  $\mzip$ as in:
\[ \xymatrix@R-0.8pc{
\natMlt[K](X)\times\natMlt[K](Y)\ar[rr]|-{\circ}^-{\arr\otimes\arr}
   \ar[d]|-{\circ}_-{\mzip} & &  X^{X}\times Y^{K}\ar[d]|-{\circ}^-{\zip}
\\
\natMlt[K](X\times Y)\ar[rr]|-{\circ}^-{\arr} & & (X\times Y)^{K}
} \]

\noindent But $\mzip$ and $\zip$ do not commute with $\acc$ instead of
$\arr$.

\item \label{MzipLemDrawDelete} Draw-and-delete commutes with
  $\mzip$, as in:
\[ \hspace*{-2em}\xymatrix@R-0.8pc@C-1pc{
\natMlt[K\!+\!1](X)\!\times\!\natMlt[K\!+\!1](Y)\ar[d]|-{\circ}_-{\mzip}
   \ar[r]|-{\circ}^-{\raisebox{4pt}{$\scriptstyle\drawdelete\otimes\drawdelete$}} & 
   \natMlt[K](X)\!\times\!\natMlt[K](Y)\ar[d]|-{\circ}^-{\mzip}
\\
\natMlt[K\!+\!1](X\!\times\! Y)\ar[r]|-{\circ}^-{\drawdelete} & 
   \natMlt[K](X\!\times\! Y)
} \]
\end{enumerate}
\end{lemma}

\begin{proof}
Most of these points are relatively straightforward, except the last
point. Recall the probabilistic projection channel $\probproj \colon
X^{K+1} \chanto X^{K}$ from the proof of
Theorem~\ref{ParMulnomLawHypergeomThm}. It commutes with $\acc$ and
$\arr$, but not with (ordinary) $\zip$. The actual projection function
$\pi\colon X^{K+1} \rightarrow X^{K}$, which removes the last element,
does commute with $\zip$. One can change appropriately from
$\probproj$ to $\pi$ since $\probproj \klafter \arr = \pi \klafter
\arr$. \QED
\end{proof}

The following result deserves a separate status. It tells that what we
learn from a zip of multisets is the same as what we learn from a
parallel product (of these multisets). On a related note we illustrate
in Section~\ref{SampleSec} that multizip can be used for compositional
parallel sampling.

\begin{theorem}
\label{MzipFlrnThm}
Multizip and frequentist learning interact well, namely as:
\[ \begin{array}{rcl}
\flrn \gg \mzip(\varphi,\psi)
& = &
\flrn(\varphi\otimes\psi).
\end{array} \]

\noindent Equivalently, in diagrammatic form:
\[ \xymatrix@R-0.8pc{
\natMlt[K](X)\times\natMlt[K](Y)\ar[d]_{\otimes}\ar[rr]|-{\circ}^-{\mzip} & & 
   \natMlt[K](X\times Y)\ar[d]|-{\circ}^-{\flrn}
\\
\natMlt[K^{2}](X\times Y)\ar[rr]|-{\circ}^-{\flrn} & & X\times Y
} \]
\end{theorem}

\begin{proof}
Let multisets $\varphi\in\natMlt[K](X)$ and $\psi\in\natMlt[K](Y)$
be given and let $a\in X$ and $b\in Y$ be arbitrary elements. We
need to show that the probability:
\[ \begin{array}{rcl}
\lefteqn{\Big(\flrn \gg \mzip(\varphi,\psi)\Big)(a,b)}
\\
& = &
\displaystyle\sum_{\vec{x}\in\acc^{-1}(\varphi)} \, \sum_{\vec{y}\in\acc^{-1}(\psi)} \,
   \frac{\acc\big(\zip(\vec{x},\vec{y})\big)(a,b)}
   {K\cdot \coefm{\varphi}\cdot\coefm{\psi}}
\end{array} \]

\noindent is the same as the probability:
\[ \begin{array}{rcl}
\flrn(\varphi\otimes\psi)(a,b)
& = &
\displaystyle\frac{\varphi(a)\cdot\psi(a)}{K\cdot K}.
\end{array} \]

\noindent We reason informally, as follows. For arbitrary
$\vec{x}\in\acc^{-1}(\varphi)$ and $\vec{y}\in\acc^{-1}(\psi)$ we need
to find the fraction of occurrences $(a,b)$ in
$\zip(\vec{x},\vec{y})$. The fraction of occurrences of $a$ in
$\vec{x}$ is $\flrn(\varphi)(a) = \frac{\varphi(a)}{K}$, and the
fraction of occurrences of $b$ in $\vec{y}$ is $\flrn(\psi)(b) =
\frac{\psi(b)}{K}$. Hence the fraction of occurrences of $(a,b)$ in
$\zip(\vec{x},\vec{y})$ is $\flrn(\varphi)(a) \cdot \flrn(\psi)(b) =
\flrn(\varphi\otimes\psi)(a,b)$. \QED
\end{proof}

Once we have seen the definition of $\mzip$, via `deconstruction' of
multisets into lists, a $\zip$ operation on lists, and
`reconstruction' to a multiset result, we can try to apply this
approach more widely. For instance, instead of using a $\zip$ on lists
we can simply concatenate $(\concat)$ the lists --- assuming they
contain elements from the same set. This yields, like in the
definition of $\mzip$, a composite channel:
\[ \vcenter{\xymatrix@R-0.8pc@C+0.2pc{
\natMlt[K](X)\times \natMlt[L](X)\ar[d]_-{\arr\otimes\arr}
   & & \Dst\big(\natMlt[K\!+\!L](X)\big)
\\
\Dst\big(X^{K}\times X^{L}\big)\ar[rr]_-{\cong}^-{\Dst(\concat)} & &
   \Dst\big(X^{K+L}\big)\ar[u]_-{\Dst(\acc)} 
}} \]

\noindent It is easy to see that this yields addition of multisets, as
a deterministic channel.

We don't get the tensor $\otimes$ of multisets in this way, because
there is no tensor of lists, because lists form a non-commutative
theory.

\section{Multizip and probabilistic operations}\label{MzipProbSec}

Having seen the essentials of the multizip operation $\mzip$ we
proceed to demonstrate how it interacts with the main probabilistic
operations of this paper, namely (parallel) multinomials and
hypergeometric channels. It turns out that multizip is a
mathematically well-behaved operation.

\begin{theorem}
\label{MzipMulnomHypergeomThm}
Multizip commutes with multinomial and hypergeometric channels, as in:
\[ \xymatrix@R-0.8pc@C+1pc{
\Dst(X)\!\times\!\Dst(Y)\ar[d]_-{\otimes}
  \ar[rr]|-{\circ}^-{\multinomial[K] \otimes \multinomial[K]} & &
  \natMlt[K](X)\!\times\!\natMlt[K](Y)\ar[d]|-{\circ}^-{\mzip}
\\
\Dst(X\!\times\! Y)\ar[rr]|-{\circ}^-{\multinomial[K]} & & \natMlt[K](X\!\times\! Y)
}\]
\[ \hspace*{-0.5em}\xymatrix@R-0.8pc@C+3pc{
\natMlt[N](X)\!\times\!\natMlt[N](Y)\ar[d]|-{\circ}_-{\mzip}
   \ar[r]|-{\circ}^-{\raisebox{5pt}{$\scriptstyle
      \hypergeometric[K]\otimes\hypergeometric[K]$}} & 
   \natMlt[K](X)\!\times\!\natMlt[K](Y)\ar[d]|-{\circ}^-{\mzip}
\\
\natMlt[N](X\!\times\! Y)\ar[r]|-{\circ}^-{\hypergeometric[K]} & 
   \natMlt[K](X\!\times\! Y)
} \]
\end{theorem}

\begin{proof}
Commutation with the hypergeometric channels follows from
Lemma~\ref{MzipLem}~\eqref{MzipLemDrawDelete}, using
Theorem~\ref{DrawDeleteHypergeometricThm}.  To prove commutation of
the multinomial rectangle, we expand the definition of $\mzip$, on the
right below.
\[ \hspace*{-0.5em}\xymatrix@R-0.8pc@C+0.7pc{
\Dst(X)\!\times\!\Dst(Y)\ar[ddd]_-{\otimes}\ar@/_2ex/[drr]_-{\iid\otimes\iid}
  \ar[rr]|-{\circ}^-{\multinomial[K] \otimes \multinomial[K]} & &
  \natMlt[K](X)\!\times\!\natMlt[K](Y)
  \ar[d]|-{\circ}^-{\arr\otimes\arr}\ar`r[d]`[dddl]|-{\circ}_-{\mzip}[ddd]
\\
& & X^{K} \!\times\! Y^{K}\ar[d]|-{\circ}^-{\zip}
\\
& & (X\!\times\! Y)^{K}\ar[d]|-{\circ}^-{\acc}
\\
\Dst(X\!\times\! Y)\ar[rr]|-{\circ}^-{\multinomial[K]}\ar@/^2ex/[urr]^-{\iid} 
   & & \natMlt[K](X\!\times\! Y)
}\]

\noindent The (triangle) subdiagrams at the top and bottom commute by
Theorem~\ref{MultinomialIIDAccThm}. The middle subdiagram is the first
rectangle in Lemma~\ref{ZipLem}. \QED
\end{proof}

\begin{remark}
\label{MulnomTensorRem}
In the first diagram in Theorem~\ref{MzipMulnomHypergeomThm} we see
that multinomial channels commute with $\otimes$ and $\mzip$. One may
wonder about the analogous diagram for tensors only, see below. This
diagram does \emph{not} commute.
\[ \xymatrix@R-1.8pc@C+1pc{
\Dst(X)\times\Dst(Y)\ar[dd]_-{\otimes}
   \ar[rr]|-{\circ}^-{\multinomial[K]\otimes\multinomial[L]} & &
   \Mlt[K](X)\times\Mlt[L](Y)\ar[dd]|-{\circ}^-{\otimes}
\\
& \neq &
\\
\Dst(X\times Y)\ar[rr]|-{\circ}_-{\multinomial[K\cdot L]} & &
  \Mlt[K\!\cdot\!L](X\times Y)
} \]

\noindent Take for instance $X = Y = \{a,b\}$, $K=1$, $L=2$ with
uniform distribution $\omega = \frac{1}{2}\ket{a} +
\frac{1}{2}\ket{b}$. The two legs of the above diagram differ when
applied to $(\omega,\omega)$.
\end{remark}

\smallskip

The \emph{pi\`ece de r\'esistance} of this paper is the commutation of
multizip $\mzip$ with the parallel multinomial law $\pml$, as
expressed by the next two results.

\begin{lemma}
\label{ParMulnomLawZipLem}
The following diagram commutes.
\[ \xymatrix@R-0.8pc@C-1.3pc{
\natMlt[K]\big(\Dst(X)\big)\!\times\!\natMlt[K]\big(\Dst(Y)\big)
   \ar[rr]|-{\circ}^-{\pml\otimes\pml}\ar[d]|-{\circ}_-{\mzip} & &
   \natMlt[K](X)\!\times\!\natMlt[K](Y)\ar[dd]|-{\circ}^-{\mzip}
\\
\natMlt[K]\big(\Dst(X)\!\times\!\Dst(Y)\big)
   \ar[d]|-{\circ}_-{\natMlt[K](\otimes)} & &
\\
\natMlt[K]\big(\Dst(X\!\times\! Y)\big)\ar[rr]|-{\circ}^-{\pml} & &
   \natMlt[K](X\!\times\! Y)
} \]
\end{lemma}

\begin{proof}
The result follows from a big diagram chase in which
the $\mzip$ operations on the left and on the right are expanded.
\[ \hspace*{-0.5em}\xymatrix@R-0.8pc@C-3.2pc{
\natMlt[K]\Dst(X)\!\times\!\natMlt[K]\Dst(Y)
   \ar[rr]|-{\circ}^-{\,\pml\otimes\pml}\ar[d]_-{\arr\otimes\arr}
   \ar`l[d]`[dddr]|-{\circ}^-{\mzip}[ddd] & &
   \natMlt[K](X)\!\times\!\natMlt[K](Y)\ar[d]|-{\circ}^-{\arr\otimes\arr}
   \ar`r[d]`[ddddl]|-{\circ}_-{\mzip}[dddd]
\\
\Dst(X)^{K}\!\times\!\Dst(Y)^{K}\ar[rr]|-{\circ}^-{\bigotimes\otimes\bigotimes}
   \ar[d]|-{\circ}_-{\zip} & &
   X^{K}\!\times\! Y^{K}\ar[d]|-{\circ}^-{\zip}
\\
\!\!\!\!\big(\Dst(X)\!\times\!\Dst(Y)\big)^{\!K}\!\ar[r]|-{\circ}^(0.62){\otimes^{K}}
   \ar[d]|-{\circ}_-{\acc} & 
   \Dst(X\!\times\! Y)^{K}\!\!\ar[r]|-{\circ}^-{\bigotimes}\ar@/^4ex/[ddl]^(0.3){\acc} & 
   (X\!\times\! Y)^{K}\ar[dd]|-{\circ}^-{\acc}
\\
\natMlt[K]\big(\Dst(X)\!\times\!\Dst(Y)\big)
   \ar[d]|-{\circ}_-{\natMlt[K](\otimes)} & &
\\
\natMlt[K]\Dst(X\!\times\! Y)\ar[rr]|-{\circ}^-{\pml} & &
   \natMlt[K](X\!\times\! Y)
} \]

\noindent The upper rectangle commutes by
Proposition~\ref{ParMulnomLawAccArrProp} and the middle one by
Lemma~\ref{ZipLem}. The lower-left subdiagram commutes by naturality
of $\acc$ and the lower-right one via the third definition of $\pml$
in~\eqref{ParMulnomLawThirdEqn}. \QED
\end{proof}

\begin{theorem}
\label{LiftedMultisetMonoidalThm}
The lifted functor $\natMlt[K] \colon \Kl(\Dst) \rightarrow \Kl(\Dst)$
from Theorem~\ref{ParMulnomLawLawThm}~\eqref{ParMulnomLawLawThmKl}
commutes with multizipping: for channels $f\colon X \chanto U$ and
$g\colon Y \chanto V$ one has:
\[ \xymatrix@R-0.8pc@C+1pc{
\natMlt[K](X)\times\natMlt[K](Y)
   \ar[d]|-{\circ}_-{\natMlt[K](f)\otimes\natMlt[K](g)}
   \ar[r]|-{\circ}^-{\mzip} &
   \natMlt[K](X\times Y)\ar[d]|-{\circ}^-{\natMlt[K](f\otimes g)}
\\
\natMlt[K](U)\times\natMlt[K](V)\ar[r]|-{\circ}^-{\mzip} & \natMlt[K](U\times V)
} \]

\noindent This channel-naturality of $\mzip$, in combination with the
unit and associativity of Lemma~\ref{MzipLem}~\eqref{MzipLemUnit}
and~\eqref{MzipLemAss}, means that the lifted functor $\natMlt[K]
\colon \Kl(\Dst) \rightarrow \Kl(\Dst)$ is monoidal via $\mzip$.
\end{theorem}

\smallskip

This result is rather subtle, since $f,g$ are used as channels. So
when we write $\natMlt[K](f)$ we mean the application of the lifted
functor $\natMlt[K] \colon \Kl(\Dst) \rightarrow \Kl(\Dst)$ to $f$,
which is:
\[ \xymatrix@C-1.5pc{
\natMlt[K](X)\ar[rrr]^-{\natMlt[K](f)} & & &
   \natMlt[K]\big(\Dst(Y)\big)\ar[rr]^-{\pml} & &
   \Dst\big(\natMlt[K](Y)\big).
} \]

\begin{proof}
We use:
\[ \hspace*{-0em}\begin{array}[b]{rcl}
\lefteqn{\mzip \klafter \Big(\natMlt[K](f) \otimes \natMlt[K](g)\Big)}
\\
& = &
\mzip \klafter (\pml\otimes\pml) \after 
   \big(\natMlt[K](f) \times \natMlt[K](g)\big)
\\
& = &
\pml \klafter \natMlt[K](\otimes) \klafter \mzip \after
   \big(\natMlt[K](f) \times \natMlt[K](g)\big)
\\
& & \qquad \mbox{by Lemma~\ref{ParMulnomLawZipLem}}
\\
& = &
\pml \klafter \natMlt[K](\otimes) \klafter
   \natMlt[K](f \times g) \klafter \mzip
\\
& & \qquad \mbox{by Lemma~\ref{MzipLem}~\eqref{MzipLemNat}}
\\
& = &
\pml \klafter \natMlt[K](f \otimes g) \klafter \mzip
\\
& = &
\natMlt[K]\big(f\otimes g\big) \klafter \mzip.
\end{array} \eqno{\QEDbox} \]
\end{proof}

\begin{remark}
\label{PmlTensorRem}
One may wonder: why bother about these $K$-sized multisets $\Mlt[K]$
and not work with $\Mlt$ instead? After all, $\Mlt$ also lifts to a
monad on $\Kl(\Dst)$, as we have seen in
Theorem~\ref{ParMulnomLawLawThm}~\eqref{ParMulnomLawLawThmMlt}. We can
use as tensor for $\Dst\Mlt$ the combination of the tensor for $\Mlt$
and for $\Dst$.

However, an important fact is that the parallel multinomial law $\pml$
does \emph{not} commute with these two tensors (of multisets and
distributions), as in the following diagram.
\[ \hspace*{-0.3em}\xymatrix@R-0.8pc@C-1.5pc{
\natMlt[K]\big(\Dst(X)\big)\!\times\!\natMlt[L]\big(\Dst(Y)\big)
   \ar[rr]|-{\circ}^-{\pml\otimes\pml}\ar[d]|-{\circ}_-{\otimes} & &
   \natMlt[K](X)\!\times\!\natMlt[K](Y)\ar[dd]|-{\circ}^-{\otimes}
\\
\natMlt[K\!\cdot\!L]\big(\Dst(X)\!\times\!\Dst(Y)\big)
   \ar[d]|-{\circ}_-{\natMlt[K\cdot L](\otimes)} & \neq &
\\
\natMlt[K\!\cdot\!L]\big(\Dst(X\!\times\! Y)\big)\ar[rr]|-{\circ}^-{\pml} & &
   \natMlt[K\!\cdot\!L](X\!\times\! Y)
} \]

\noindent Take for instance $X = \{a,b\}$, $Y = \{0,1\}$ with $K=2$,
$L=1$ and with multisets $\varphi = 2\bigket{\frac{3}{4}\ket{a} +
  \frac{1}{4}\ket{b}}\in\Mlt[K](\Dst(X))$ and $\psi =
1\bigket{\frac{2}{3}\ket{0} + \frac{1}{3}\ket{1}}\in\Mlt[L](\Dst(Y))$.
With some perseverance one sees that the two legs of the above
diagrams give different outcomes on $(\varphi,\psi)$.

This non-commutation generalises what we have seen in
Remark~\ref{MulnomTensorRem} for multinomial channels.
\end{remark}

\begin{theorem}
\label{LiftedMultisetPlusThm}
The lifted functors $\natMlt[K] \colon \Kl(\Dst) \rightarrow
\Kl(\Dst)$ commute with sums of multisets: for a channel $f\colon X
\chanto Y$,
\[ \xymatrix@R-0.8pc@C+1pc{
\natMlt[K](X)\times\natMlt[L](Y)
   \ar[d]|-{\circ}_-{\natMlt[K](f)\otimes\natMlt[L](f)}
   \ar[r]|-{\circ}^-{+} &
   \natMlt[K\!+\!L](X)\ar[d]|-{\circ}^-{\natMlt[K+L](f)}
\\
\natMlt[K](Y)\times\natMlt[L](Y)\ar[r]|-{\circ}^-{+} & \natMlt[K\!+\!L](Y)
} \]
\end{theorem}

\begin{proof}
By Proposition~\ref{ParMulnomLawPlusProp}, in combination with the
standard fact that the multiset functor $\Mlt\colon\Sets\rightarrow
\Sets$ commutes with sums of multisets. \QED
\end{proof}

Now that we know, via
Theorem~\ref{ParMulnomLawLawThm}~\eqref{ParMulnomLawLawThmKl}, that
$K$-sized multisets form a (lifted) functor $\natMlt[K] \colon
\Kl(\Dst) \rightarrow \Kl(\Dst)$, we can ask whether the natural
transformations that we have used before are also natural with respect
to this lifted functor. This involves naturality with respect to
channels instead of with respect to ordinary functions. The proofs of
the following result are left to the reader.

\begin{lemma}
\label{ArrDDChanNatLem}
Arrangement and accumulation, and draw-delete are natural
transformation in the situations:
\[ \xymatrix@C-0.0pc{
\Kl(\Dst)\ar@/^4ex/[rr]^-{\natMlt[K]}_-{\big\Downarrow\rlap{$\scriptstyle\arr$}}
   \ar[rr]^(0.32){(-)^{K}}
   \ar@/_4ex/[rr]_-{\natMlt[K]}^-{\big\Downarrow\rlap{$\scriptstyle\acc$}}
   & & \Kl(\Dst)\hspace*{-1em}
&
\hspace*{-1em}\Kl(\Dst)\ar@/^2.2ex/[rr]^-{\natMlt[K+1]}\ar@/_2.2ex/[rr]_-{\natMlt[K]} &
  {\big\Downarrow}\rlap{$\scriptstyle\drawdelete$} & \Kl(\Dst)
} \]

\noindent The power functor $(-)^{K}$ on $\Sets$ lifts to $\Kl(\Dst)$
via the big tensor $\bigotimes\colon \Dst(X)^{K} \rightarrow
\Dst(X^{K})$. \QED
\end{lemma}

The distribution functor $\Dst$ on $\Sets$ lifts in a standard way to
a (monoidal) functor on its Kleisli category $\Kl(\Dst)$, also written
as $\Dst$. This is used in the following channel-naturality result for
multinomial and hypergeometric distributions.

\begin{theorem}
The multinomial and hypergeometric channels are natural with respect
to lifted functors:
\[ \xymatrix@C-0.0pc{
\Kl(\Dst)\ar@/^2.2ex/[rr]^-{\Dst}_-{\big\Downarrow\rlap{$\scriptstyle\multinomial[K]$}}
   \ar@/_2.2ex/[rr]_-{\natMlt[K]}
   & & \Kl(\Dst)\hspace*{-0.5em}
&
\hspace*{-0.5em}\Kl(\Dst)\ar@/^2.2ex/[rr]^-{\natMlt[L]}\ar@/_2.2ex/[rr]_-{\natMlt[K]} &
  {\big\Downarrow}\rlap{$\scriptstyle\hypergeometric[K]$} & \Kl(\Dst)
} \]

\noindent where $L\geq K$. These natural transformations are monoidal
since they commute with $\mzip$, see
Theorem~\ref{MzipMulnomHypergeomThm}. \QED
\end{theorem}

This result shows that the basic probabilistic operations of drawing
from an urn, with or without replacement, come with rich categorical
structure.

\section{Application in sampling semantics}\label{SampleSec}

The above theory about the distributive law of multisets over
distributions can be used to demonstrate the correctness of (basic
aspects of) sampling semantics in probabilistic programming languages
(see \textit{e.g.}~\cite{GordonHNRG14,ScibiorGG15,BartheKS20} for an
overview). We concentrate on parallel and sequential composition and
on conditioning.

We formalise sampling via the multinomial channel $\multinomial[K]
\colon \Dst(X) \chanto \Mlt[K](X)$. It assigns probabilities to
samples of size $K$, from a distribution $\omega\in\Dst(X)$. The fact
that $\flrn \gg \multinomial[K](\omega) = \omega$, as we have seen in
Proposition~\ref{MultinomialSumFlrnProp}~\eqref{MultinomialSumFlrnPropFlrn},
shows that learning from these samples yields the original
distribution $\omega$ --- a key correctness property for sampling.

The first diagram in Theorem~\ref{MzipMulnomHypergeomThm} shows
that sampling a product distribution:
\[ \multinomial[K](\omega\otimes\rho) \]

\noindent is the same as the mzip of parallel (separate) sampling:
\[ \mzip \gg \Big(\multinomial[K](\omega) \otimes \multinomial[K](\rho)\Big). \]

\noindent This gives a basic compositionality result in sampling
semantics. At the same time it illustrates the usefulness of $\mzip$.
Applying frequentist learning to these expressions yields the original
product distribution $\omega\otimes\rho$, see also
Theorem~\ref{MzipFlrnThm}.

We continue with sequential sampling.
\begin{quote}
Suppose we have a distribution $\omega\in\Dst(X)$ and a channel $c
\colon X \chanto Y$. We wish to sample the distribution $c \gg \omega
\in \Dst(Y)$ obtained by state transformation. A common way to do so
is to first sample elements $x$ from $\omega$ and then elements $y$
from $c(x)$. These $y$'s are the samples that capture $c \gg \omega$.
\end{quote}

We can first sample from $\omega$ via
$\multinomial[K](\omega)\in\Mlt[K](X)$. Subsequently we can apply the
channel $c$ to each of the elements in the sample, via $\Mlt[K](c)
\colon \Mlt[K](X) \rightarrow \Mlt[K](\Dst(Y))$. Our next step is to
sample from all of these outcomes $c(x)$. But this is precisely what
the parallel multinomial law $\pml$ does. This gives the samples
(multisets) of $Y$ that we are after, via the composite of channels:
\[ \hspace*{-0.4em}\xymatrix@C-1.1pc{
\Dst(X)\ar[rr]|-{\circ}^-{\multinomial[K]} & &
   \Mlt[K](X)\ar[rr]|-{\circ}^-{\Mlt[K](c)} & &
   \Mlt[K](\Dst(Y))\ar[r]|-{\circ}^-{\pml} & \Mlt[K](Y)
} \]

\noindent The last two arrows describe the lifted functor $\Mlt[K]$
applied to the channel $c\colon X \chanto Y$.  We claim that the next
result is a correctness result for the above sampling method for $c
\gg \omega$.

\begin{proposition}
\label{SampleProp}
For a distribution $\omega\in\Dst(X)$ and channel $c\colon X\chanto Y$,
\[ \begin{array}{rcl}
\flrn \gg \Big(\pml \klafter \Mlt[K](c) \klafter 
   \multinomial[K]\Big)(\omega)
& = &
c \gg \omega.
\end{array} \]
\end{proposition}

\begin{proof}
Because:
\[ \hspace*{-3em}\begin{array}[b]{rcl}
\lefteqn{\big(\flrn \klafter \pml \klafter \Mlt[K](c) \klafter 
   \multinomial[K]\big)(\omega)}
\\
& = &
\big(\idmap \klafter \flrn \klafter \Mlt[K](c) \klafter 
   \multinomial[K]\big)(\omega)
   \rlap{ by Theorem~\ref{ParMulnomLawFlrnThm}}
\\
& = &
\big(\idmap \klafter \Dst(c) \klafter \flrn \klafter 
   \multinomial[K]\big)(\omega)
   \quad \mbox{by naturality}
\\
& = &
\big(\idmap \klafter \Dst(c)\big)(\omega)
   \quad \mbox{by Proposition~\ref{MultinomialSumFlrnProp}~\eqref{MultinomialSumFlrnPropFlrn}}
\\
& = &
\big(\mu \after \Dst(c)\big)(\omega)
\\
& = &
c \gg \omega.
\end{array} \eqno{\QEDbox} \]
\end{proof}

Updating (or conditioning) of a distribution is a basic operation in
probabilistic programming. We briefly sketch how it fits in the
current (sampling) setting. Let $\omega\in\Dst(X)$ be distribution and
$p\colon X \rightarrow [0,1]$ be a (fuzzy) predicate. We write
$\omega\models p \coloneqq \sum_{x} \omega(x)\cdot p(x)$ for the
validity (expected value) of $p$ in $\omega$. If this validity is
non-zero, we can define the updated distribution $\omega|_{p}
\in\Dst(X)$ as the normalised product:
\[ \begin{array}{rcl}
\omega|_{p}(x)
& \coloneqq &
\displaystyle\frac{\omega(x)\cdot p(x)}{\omega\models p}.
\end{array} \]

\noindent See \textit{e.g.}~\cite{Jacobs19b,Jacobs19c,Jacobs21a} for
more information.

The question arises: how to sample from an updated state
$\omega|_{p}$? This involves the interaction of updating with
(parallel) multinomials.  It is addressed below. There we use the free
extension $\overline{p} \colon \natMlt(X) \rightarrow [0,1]$ of a
predicate $p\colon X \rightarrow [0,1]$ to multisets, via conjunction
$\andthen$, that is, via pointwise multiplication:
\[ \begin{array}{rcl}
\overline{p}(\varphi)
& \coloneqq &
\displaystyle\prod_{x\in X} \, 
   \underbrace{p(x) \cdot \ldots \cdot p(x)}_{\varphi(x)\text{ times}}.
\end{array} \]

\noindent Point~\eqref{MultinomialUpdateThmUpd} below shows that
sampling of an updated distribution $\omega|_{p}$ can be done
compositionally, by first sampling $\omega$ and then updating the
samples with the free extension $\overline{p}$. If $p$ is a sharp
(non-fuzzy) predicate, this means throwing out the sampled multisets
where $p$ does not hold for all elements, followed by
re-normalisation.

\begin{theorem}
\label{MultinomialUpdateThm}
For a distribution $\omega\in\Dst(X)$, a predicate $p\colon X
\rightarrow [0,1]$, and a number $K\in\NNO$, one has:
\begin{enumerate}
\item \label{MultinomialUpdateThmVal} $\multinomial[K](\omega) \models
  \overline{p} \,=\, \big(\omega\models p\big)^{K}$;

\item \label{MultinomialUpdateThmUpd}
  $\multinomial[K](\omega)\big|_{\overline{p}} =
  \multinomial[K](\omega|_{p})$;

\item \label{MultinomialUpdateThmPmlVal} Similarly, for the parallel
  multinomial law $\pml$,
\[ \begin{array}{rcl}
\pml\big(\sum_{i}n_{i}\ket{\omega_i}\big)\models \overline{p}
& \,=\, &
{\displaystyle\prod}_{i} \, \big(\omega_{i}\models p\big)^{n_i}.
\end{array} \]

\item \label{MultinomialUpdateThmPmlUpd} And:
\[ \begin{array}{rcl}
\pml\big(\sum_{i}n_{i}\ket{\omega_i}\big)\big|_{\overline{p}}
& \,=\, &
\pml\big(\sum_{i}n_{i}\ket{\,\omega_{i}|_{p}}\big).
\end{array} \]
\end{enumerate}
\end{theorem}

An example that combines multinomials and updating occurs
in~\cite[\S6.4]{Ross18}, but without a general formulation,
as given above, in point~\eqref{MultinomialUpdateThmUpd}.


\begin{proof}
The first point follows from the Multinomial Theorem. Then,
for $\varphi\in\natMlt[K](X)$,
\[ \begin{array}{rcl}
\multinomial[K](\omega)\big|_{\overline{p}}(\varphi)
& = &
\displaystyle
   \frac{\multinomial[K](\omega)(\varphi) \cdot \overline{p}(\varphi)}
        {\multinomial[K](\omega) \models \overline{p}}
\\[+1em]
& \smash{\stackrel{\eqref{MultinomialUpdateThmVal}}{=}} &
\displaystyle \coefm{\varphi} \cdot
   \frac{\prod_{x} \omega(x)^{\varphi(x)}\cdot \prod_{x} p(x)^{\varphi(x)}}
        {(\omega\models p)^{K}}
\\[+0.8em]
& = &
\displaystyle \coefm{\varphi} \cdot
   \frac{\prod_{x} \big(\omega(x)\cdot p(x)\big)^{\varphi(x)}}
        {(\omega\models p)^{K}}
\\[+0.8em]
& = &
\displaystyle \coefm{\varphi} \textstyle \cdot
   {\displaystyle\prod}_{x}\, \displaystyle
   \left(\frac{\omega(x)\cdot p(x)}{\omega\models p}\right)^{\varphi(x)}
\\[+1em]
& = &
\displaystyle \coefm{\varphi} \textstyle \cdot
   {\displaystyle\prod}_{x}\, \omega|_{p}(x)^{\varphi(x)}
\\[+0.8em]
& = &
\multinomial[K](\omega|_{p})(\varphi).
\end{array} \]

\noindent For point~\eqref{MultinomialUpdateThmPmlVal} we use the second
formulation~\eqref{ParMulnomLawSecondEqn} of $\pml$ in:
\[ \begin{array}{rcl}
\lefteqn{\pml\big(\sum_{i}n_{i}\ket{\omega_i}\big)\models\overline{p}}
\\
& = &
\displaystyle\sum_{i,\varphi_{i}\in\natMlt[n_i](X)} \textstyle
   \Big(\prod_{i}\, \multinomial[n_{i}](\omega_{i})(\varphi_{i})\Big)\cdot
   \overline{p}\big(\sum_{i}\varphi_{i}\big)
\\[+1.3em]
& = &
\displaystyle\sum_{i,\varphi_{i}\in\natMlt[n_i](X)} \textstyle
   \Big(\prod_{i}\, \multinomial[n_{i}](\omega_{i})(\varphi_{i})\Big)\cdot
   \Big(\prod_{i} \overline{p}(\varphi_{i})\Big)
\\[+0.8em]
& & \qquad \mbox{since $\overline{p}$ is by definition a map of monoids}
\\[+0.5em]
& = &
\displaystyle\sum_{i,\varphi_{i}\in\natMlt[n_i](X)} \textstyle
   \prod_{i}\, \multinomial[n_{i}](\omega_{i})(\varphi_{i}) 
   \cdot \overline{p}(\varphi_{i})
\\[+1.3em]
& = &
{\displaystyle\prod}_{i} \, 
   \displaystyle\sum_{\varphi_{i}\in\natMlt[n_i](X)} \, 
   \multinomial[n_{i}](\omega_{i})(\varphi_{i}) \cdot \overline{p}(\varphi_{i})
\\[+1.3em]
& = &
{\displaystyle\prod}_{i} \, 
   \multinomial[n_{i}](\omega_{i}) \models \overline{p}
\\[+0.8em]
& \smash{\stackrel{\eqref{MultinomialUpdateThmVal}}{=}} &
{\displaystyle\prod}_{i} \, \big(\omega_{i} \models p\big)^{n_i}.
\end{array} \]

\noindent This formula for validity is used to prove
the final equation about updating. Let $K = \sum_{i}n_{i}$ in:
\[ \hspace*{-2.5em}\begin{array}[b]{rcl}
\lefteqn{\textstyle\pml\big(\sum_{i}n_{i}\ket{\omega_i}\big)\big|_{\overline{p}}}
\\[+0.5em]
& = &
\displaystyle\sum_{\varphi\in\natMlt[K](X)}
   \frac{\pml\big(\sum_{i}n_{i}\ket{\omega_i}\big)(\varphi) \cdot 
         \overline{p}(\varphi)}
        {\pml\big(\sum_{i}n_{i}\ket{\omega_i}\big) \models \overline{p}}
   \,\bigket{\varphi}
\\[+1.3em]
& = &
\displaystyle\sum_{i,\varphi_{i}\in\natMlt[n_i](X)} 
   \frac{\prod_{i}\, \multinomial[n_{i}](\omega_{i})(\varphi_{i})\cdot
   \overline{p}(\varphi_{i})}
   {\prod_{i}\, (\omega_{i}\models p)^{n_i}} \textstyle
   \,\bigket{\sum_{i}\varphi_{i}}
\\[+1em]
& = &
\displaystyle\sum_{i,\varphi_{i}\in\natMlt[n_i](X)} \textstyle
   {\displaystyle\prod}_{i}\; \coefm{\varphi_i}\cdot \displaystyle
   \frac{\prod_{x}\, (\omega(x)\cdot p(x))^{\varphi_{i}(x)}}
   {(\omega_{i}\models p)^{n_i}} \textstyle
   \,\rlap{$\bigket{\sum_{i}\varphi_{i}}$}
\\[+1em]
& = &
\displaystyle\sum_{i,\varphi_{i}\in\natMlt[n_i](X)} \textstyle
   {\displaystyle\prod}_{i} \coefm{\varphi_i}\cdot 
   {\displaystyle\prod}_{x} \displaystyle\!
   \left(\frac{\omega(x)\cdot p(x)}{\omega_{i}\models p}\right)^{\!\varphi_{i}(x)}
   \textstyle\!\rlap{$\bigket{\sum_{i}\varphi_{i}}$}
\\[+1.3em]
& = &
\displaystyle\sum_{i,\varphi_{i}\in\natMlt[n_i](X)} \textstyle
   {\displaystyle\prod}_{i}\; \coefm{\varphi_i}\cdot 
   {\displaystyle\prod}_{x} \, \omega_{i}|_{p}(x)^{\varphi_{i}(x)}
   \,\bigket{\sum_{i}\varphi_{i}}
\\[+1.3em]
& = &
\displaystyle\sum_{i,\varphi_{i}\in\natMlt[n_i](X)} \textstyle
   {\displaystyle\prod}_{i}\; \multinomial[n_{i}](\omega_{i}|_{p})(\varphi_{i})
   \,\bigket{\sum_{i}\varphi_{i}}
\\[+1.3em]
& = &
\pml\big(\sum_{i}n_{i}\ket{\,\omega_{i}|_{p}}\big).
\end{array} \eqno{\QEDbox} \]

\auxproof{
\[ \begin{array}{rcl}
\multinomial[K](\omega) \models \overline{p}
& = &
\displaystyle\sum_{\varphi\in\natMlt[K](X)}
   \multinomial[K](\omega)(\varphi) \cdot \overline{p}(\varphi)
\\[+1.3em]
& = &
\displaystyle\sum_{\varphi\in\natMlt[K](X)} \coefm{\varphi} \textstyle \cdot
   \left({\displaystyle\prod}_{x}\, \omega(x)^{\varphi(x)}\right) \cdot
   \left({\displaystyle\prod}_{x}\, p(x)^{\varphi(x)}\right)
\\[+1.3em]
& = &
\displaystyle\sum_{\varphi\in\natMlt[K](X)} \coefm{\varphi} \textstyle \cdot
   {\displaystyle\prod}_{x}\, \big(\omega(x)\cdot p(x)\big)^{\varphi(x)}
\\[+1.3em]
& \smash{\stackrel{\eqref{MultinomTheoremEqn}}{=}} &
\big(\sum_{x} \omega(x)\cdot p(x)\big)^{K}
\\
& = &
\big(\omega\models p\big)^{K}.
\end{array} \]
}
\end{proof}

\section{Final remarks}\label{ConclSec}

This paper has given an exposition on the distributivity of multisets
over probability distributions, via what has been named the parallel
multinomial law. It demonstrates (once more) that basic results in
probability theory find a natural formulation in the categorical
framework of channels / Kleisli-maps. Highlights of this paper are the
commutation of the parallel multinomial law with hypergeometric
distributions (Theorem~\ref{ParMulnomLawHypergeomThm}), and the
multizip operation that interacts smoothly with frequentist learning,
(parallel) multinomials and hypergeometrics, see
Theorems~\ref{MzipFlrnThm} and~\ref{MzipMulnomHypergeomThm} and
Lemma~\ref{ParMulnomLawZipLem}. The results make the $K$-sized
multiset functor on $\Kl(\Dst)$ monoidal and also make various natural
transformations monoidal. This may enrich axiomatic approaches to
probability like in~\cite{Fritz20}. The channel-based approach could
also be of more practical interest for popular probabilistic
computation tools like SPSS or R.

The probabilistic channels in this paper have been implemented in
Python (via the EfProb library~\cite{ChoJ17}). This greatly helped in
grasping what's going on, and for testing results (with random
inputs). The commuting diagrams (of channels) in this paper may look
simple, but the computations involved are quite complicated and
quickly consume many resources (and then run out of memory or take a
long time).

This work offers ample room for extension, for instance with P\'olya
channels (like in~\cite{JacobsS20}) or to continuous probability (like
in~\cite{DashS20}).





\end{document}